%% file: alt_tvp_main.tex
\documentclass[12pt,a4,notitlepage,fleqn]{article}

\usepackage[verbose]{geometry}
\usepackage[symbol]{footmisc}
\usepackage[T1]{fontenc}
\usepackage{authblk}



\title{An Alternative Estimation Method of a Time-Varying Parameter Model}%

\author{Mikio Ito$^{a}$, \ Akihiko Noda$^{b,c}$ \ and \ Tatsuma Wada$^{d}$\thanks{\scriptsize Corresponding Author. E-mail: twada@sfc.keio.ac.jp,
Tel. +81-466-49-3451, Fax. +81-466-49-3451}

{\scriptsize ${}^{a}$ \it Faculty of Economics, Keio University, 2-15-45 Mita, Minato-ku, Tokyo 108-8345, Japan}

{\scriptsize ${}^{b}$ \it Faculty of Economics, Kyoto Sangyo University, Motoyama, Kamigamo, Kita-ku, Kyoto 603-8555, Japan}

{\scriptsize ${}^{c}$ \it Keio Economic Observatory, Keio University, 2-15-45 Mita, Minato-ku, Tokyo 108-8345, Japan}

{\scriptsize ${}^{d}$ \it Faculty of Policy Management, Keio University, 5322 Endo, Fujisawa, Kanagawa, 252-0882, Japan}}

\date{This Version: \today}

\renewcommand\thefootnote{\arabic{footnote}}

\usepackage{graphicx} 

\usepackage[]{natbib}%
\usepackage{amsmath,amssymb}%
\usepackage{ascmac}%
\usepackage{multirow}%
\usepackage{lscape}%
\usepackage{subfigmat}

\usepackage{pifont}%
\usepackage{arydshln}%
\usepackage[format=hang]{caption}
\usepackage[all]{xy}
\usepackage{url}
\bibpunct{(}{)}{;}{a}{}{,}

\def\hsymbu#1{\smash{\lower1.7ex\hbox{\huge$#1$}}}

\newtheorem{lemma}{Lemma}
\newtheorem{proposition}{Proposition}

\newcommand{\citetapos}[1]{\citeauthor{#1}'s \citeyearpar{#1}}

\newenvironment{proof}[1][Proof]{\noindent\textbf{#1.} }{\ \rule{0.5em}{0.5em}}

\begin{document}
\begin{titlepage}

\renewcommand{\thepage}{}
\renewcommand{\thefootnote}{\fnsymbol{footnote}}

\maketitle

\vspace{-10mm}

\noindent
\hrulefill

\noindent
{\bfseries Abstract:} A non-Bayesian, regression-based or generalized least squares (GLS)-based approach is formally proposed to estimate a class of time-varying AR parameter models. This approach has partly been used by \citet{ito2014ism,ito2016eme,ito2016tvc}, and is proven to be efficient because, unlike conventional methods, it does not require Kalman filtering and smoothing procedures, but yields a smoothed estimate that is identical to the Kalman-smoothed estimate. Unlike the maximum likelihood estimator, the possibility of the pile-up problem is negligible. In addition, this approach enables us to deal with stochastic volatility models, models with a time-dependent variance-covariance matrix, and models with non-Gaussian errors that allow us to deal with abrupt changes or structural breaks in time-varying parameters.\\


 
\noindent
\hrulefill

\end{titlepage}

\bibliographystyle{asa}%

\pagebreak

\input{alt_tvp_intro}

\input{alt_tvp_model}

\input{alt_tvp_est_tvar}

\input{alt_tvp_simulation}

\input{alt_tvp_application}

\input{alt_tvp_conclusion}

\input{alt_tvp_ack}

\clearpage

\input{alt_tvp_main.bbl}
\clearpage

\input{alt_tvp_table}

\clearpage

\input{alt_tvp_appendix}

\end{document}

%% file: alt_tvp_intro.tex
\section{Introduction}\label{alt_tvp_sec1}

It has been widely recognized among Macroeconomists that the time-varying parameter models are flexible enough to capture the complex nature of a macroeconomic system, thereby yielding better forecasts and a better fit to data than models with constant parameters. In the literature on dynamic econometrics models, the instability of the parameters in the model has been often incorporated in Markov-switching models (e.g., \citet{hamilton1989nae}) or structural change models (e.g., \citet{perron1989tgc}). However, time-varying models allow the parameters to change gradually over time, which is the main difference between time-varying models and Markov switching or structural break models.

In the literature on the application of time-varying vector autoregressive (TV-VAR) models to macroeconomics, \citet{bernanke1998mmp} consider the possibility of autoregressive parameters being time-varying. However, after confirming the stability of the parameters using the parameter consistency test of \citet{hansen1992a}, they employ the time-invariant (usual) VAR model. Regarding this modeling strategy, \citet{cogley2005dvm} find that \citetapos{hansen1992a} test has low power and is unreliable. They instead propose a TV-VAR model with stochastic volatility in the error term. A study by \citet{primiceri2005tvs} sheds light on a technical aspect of the time-varying model, particularly, the Bayesian estimation technique for the time-varying parameters. In general, difficulties in dealing with time-varying parameter models arise when free parameters and unobserved variables need to be estimated. \citet{primiceri2005tvs} presents a clear estimation procedure based on the Bayesian Markov Chain Monte Carlo (MCMC) method.

Several studies, including \citet{primiceri2005tvs}, claim that the Bayesian method is preferred to the maximum likelihood (ML) method because the former (i) is less likely to suffer from the so-called pile-up problem (\citet{sargan1983mle}); (ii) is less likely to have computation problems, such as a degenerated likelihood function or multiple local minima; and (iii) facilitate the finding of statistical inferences such as standard errors. However, both the Bayesian and ML methods require Kalman filtering to estimate an unobservable state vector that includes the time-varying parameters.\footnote{An alternative to those two methods is \citet{cooley1976eps}, who do not utilize Kalman filtering but employ the likelihood method to estimate unknown parameters.}

Attempts to understand Kalman filtering through the lens of conventional regression literature are given, for example, by \citet{duncan1972ldr}, \citet{maddala1998urc}, and \citet{durbin2012tsa}. To our knowledge, \citet{duncan1972ldr} is the first study to show that the generalized least squares (GLS) estimator for basic state-space models equivalently uncovers the unobserved state vector estimated through Kalman filtering. Similarly, a series of papers by \citet{ito2014ism,ito2016eme,ito2016tvc} apply the TV-VAR, time-varying autoregressive (TV-AR), and time-varying vector error correction (TV-VEC) models to stock prices and exchange rates, without using the Kalman filter but using the regression method. In this paper, following the spirit of \citet{duncan1972ldr}, we elucidate the statistical properties of the regression-based approach or the GLS-based approach that utilizes ordinary least squares (OLS) or GLS in lieu of the Kalman smoother. To be more precise, hereafter, our GLS-based approach includes OLS as a variety of GLS. In recent studies, this approach is employed by \citet{ito2014ism,ito2016eme,ito2016tvc} to evaluate market efficiency in stock markets and foreign exchange markets.\footnote{Note that \citet{ito2014ism,ito2016eme} do not formally prove that their regression-based approach generates estimates that are equivalent to the Kalman-smoothed estimates.} In this paper, we first show that the class of TV-AR models, which includes the TV-AR, TV-VAR, and TV-VEC models, is readily estimated using the GLS method. The estimates are, in fact, tantamount to the Kalman-smoothed estimates. This finding may not be surprising given \citet{duncan1972ldr} or extensions thereof. Additionally, one may argue that the main purpose of employing the Kalman filter (or smoother) is to avoid using a system of large matrices required by GLS. This argument was reasonably strong until computers became capable of handling large matrices.

The equivalence between GLS and the Kalman smoother leads us to the following question: If GLS yields the Kalman-smoothed estimates, then how good is the GLS-based approach in recovering time-varying parameters? This question is practical and important because, in general, finite sample properties of the GLS estimator are unknown.\footnote{Note that we have unknown parameters such as the variances of error terms in our model. In such a case, we have to rely on Feasible GLS (FGLS), which may not be equivalent to GLS.} Another question pertains to the seriousness of the pile-up problem. The pile-up problem is said to occur when the ML estimate of the variance of the state equation error is zero, even though its true value is small but not zero. While our proposed method is not identical to ML because we do not maximize the likelihood function with respect to the variances of errors, it is not immediately obvious whether our GLS-based approach suffers from the pile-up problem to the same degree as ML.

Also considered are the possibilities of non-independent and identically distributed (i.i.d.) or non-Gaussian errors in the model. The former is repeatedly used in this area of study because it is reasonable to assume that the variance of errors has a variance that may be time-varying. The latter is important in empirical studies because it allows us to model abrupt changes or structural breaks in time-varying parameters, which is a similar strategy employed by \citet{perron2009ltb} and elsewhere.

To sum up, the contributions of the study are the following: We present the equivalence of Kalman smoothing and GLS for the class of TV-AR models. We then show that the GLS estimates the true time-varying parameters fairly well even when the errors are not i.i.d. or not Gaussian, provided an appropriate way to implement feasible GLS (FGLS) is carefully chosen based primarily on the relative size of the variances of the errors or signal-to-noise ratio (SNR). The pile-up problem that is often cumbersome to ML is shown to be negligible.

The rest of this paper is organized as follows. Section \ref{alt_tvp_sec2} presents our model together with its likelihood function. We analyze the statistic properties of the GLS-based approach for the class of TV-AR models in Section \ref{alt_tvp_sec3}. Section \ref{alt_tvp_sec4} evaluates the GLS-based approach under a variety of conditions such as a small SNR, non i.i.d. errors, and non-Gaussian errors. An application to macroeconomic data, including a comparison with the Bayesian MCMC method, is demonstrated in Section \ref{alt_tvp_sec5}. Section \ref{alt_tvp_sec6} concludes the paper.

%% file: alt_tvp_model.tex
\section{Model}\label{alt_tvp_sec2}
In this section, we present our model, which admits the class of TV-AR models. We then show that our model permits two different matrix forms. The first matrix form is that of \citet{durbin2012tsa}, which they use as a device to find the Kalman-smoothed estimate of an unobserved state vector. The second matrix form is an extended version of \citet{maddala1998urc}, which we employ in this paper. As it will become clear, this form allows us to use GLS for the estimation of time-varying parameters. We can then formally demonstrate that the Kalman-smoothed estimate from the first matrix form is equivalent to the GLS estimates of the second matrix form, proving that GLS estimates are an alternative estimation method to the Kalman smoother.

\subsection{A Basic State-Space Model of the Class of Time-Varying AR Models}
Our model is given by: 
\begin{eqnarray}
y_{t}&=&Z_{t}\beta_{t}+\varepsilon_{t}  \label{model1} \\
\beta_{t}&=&\beta_{t-1}+\eta_{t},  \label{model2}
\end{eqnarray}
where $y_{t}$ is a $k\times 1$ vector of observable variables; $Z_{t}$ is a $k\times m$ matrix of observable variables; $\beta_{t}$ is an $m\times 1$ vector of time-varying parameters; and $\varepsilon_{t}$ and $\eta_{t}$ are $k\times 1$ and $m\times 1$ vectors of normally distributed error terms with variance covariance matrices of $H_{t}$ and $Q_{t}$, respectively: 
\begin{equation*}
\left[ 
\begin{array}{c}
\varepsilon_{t} \\ 
\eta_{t}%
\end{array}%
\right] \sim \mathcal{N}\left( \left[ 
\begin{array}{c}
0 \\ 
0%
\end{array}%
\right] ,\left[ 
\begin{array}{cc}
H_{t} & 0 \\ 
0 & Q_{t}%
\end{array}%
\right] \right) .
\end{equation*}%
Note that the variance-covariance matrices $H_{t}$ and $Q_{t}$ are allowed to be time dependent, as in the stochastic volatility model. For the initial value of $\beta_{t}$, we assume%
\begin{equation*}
\beta_{0}\sim \mathcal{N}\left(b_{0},P_{0}\right) .
\end{equation*}%
If we assume $b_{0}$ and $P_{0}$ are known, it is reasonable to utilize the diffuse prior for $P_{0}$ because $\beta_{t}$ follows a non-stationary process. In this case, the diagonal elements of $P_{0}$ should be large numbers (e.g., see \citet{harvey1989fst}; \citet{koopman1997eik}). Alternatively, we can simply ignore $P_{0}$ as zero when we assume $\beta_{0}$ is known and not stochastic.

Equations (\ref{model1}) and (\ref{model2}) can be utilized for a variety of TV-AR models. For example, when $k=1$, $Z_{t}=y_{t-1}$ yields a TV-AR(1) model. Similarly, the TV-VAR(1) model $y_{t}=A_{t}y_{t-1}+\varepsilon_{t}$ \ with $A_{t}=A_{t-1}+\eta_{t}$ is expressed by setting $Z_{t}=\left(y_{t-1}^{\prime }\otimes I_{k}\right)$ and $\beta_{t}=vec\left(A_{t}\right)$. It is also possible to include intercepts that vary with time. For a TV-AR(1) model, for example, one can set $Z_{t}=\left(1,y_{t-1}\right)$, and then, the first element of $\beta_{t}$ is the time-varying intercept.

We present two specifications of our model, (\ref{model1}) and (\ref{model2}), below. The first specification allows us to derive the Kalman-smoothed estimate as explained by \citet{durbin2012tsa}. The second specification is in the same spirit as \citet{duncan1972ldr}, leading us to the GLS-based approach. As we shall see, both specifications yield the same smoothed estimate.

\subsection{Model Matrix Formulation of the State-Space Model}

Following \citet{durbin2012tsa}, we employ the matrix formulation of equations (\ref{model1}) and (\ref{model2}). For $t=1,\ldots ,T$, we have a system of equations:
\begin{eqnarray}
Y_{T} &=&Z\beta +\varepsilon  \label{mat1} \\
\beta &=&C\left( b_{0}^{\ast }+\eta \right)  \label{mat2}
\end{eqnarray}%
where
\begin{equation*}
\varepsilon \sim N\left( 0,H\right) ,\text{ \ \ }\eta \sim N\left(
0,Q\right) ,\text{ \ }
\end{equation*}%
with 
\begin{eqnarray*}
Y_{T} &=&\left[ 
\begin{array}{c}
y_{p+1} \\ 
y_{p+2} \\ 
\vdots \\ 
y_{T}%
\end{array}%
\right] ,\text{ \ }Z=\left[ 
\begin{array}{cccc}
Z_{p+1} &  &  & 0 \\ 
& Z_{p+2} &  &  \\ 
&  & \ddots &  \\ 
0 &  &  & Z_{T}%
\end{array}%
\right] ,\text{ \ }\beta =\left[ 
\begin{array}{c}
\beta _{p+1} \\ 
\beta_{p+2} \\ 
\vdots \\ 
\beta_{T}%
\end{array}%
\right] ,\text{ \ }\varepsilon =\left[ 
\begin{array}{c}
\varepsilon_{p+1} \\ 
\varepsilon_{p+2} \\ 
\vdots \\ 
\varepsilon_{T}%
\end{array}%
\right] \\
H &=&\left[ 
\begin{array}{cccc}
H_{p+1} &  &  & 0 \\ 
& H_{p+2} &  &  \\ 
&  & \ddots &  \\ 
0 &  &  & H_{T}%
\end{array}%
\right] ,\text{ \ }C=\left[ 
\begin{array}{cccc}
I & 0 & \cdots & 0 \\ 
I & I &  & \vdots \\ 
\vdots & \vdots & \ddots & 0 \\ 
I & I & I & I%
\end{array}%
\right] ,\text{ } \\
\text{\ }b_{0}^{\ast } &=&\left[ 
\begin{array}{c}
b_{0} \\ 
0 \\ 
\vdots \\ 
0%
\end{array}%
\right] ,\text{ \ }P_{0}^{\ast }=\left[ 
\begin{array}{cccc}
P_{0} & 0 &  & 0 \\ 
0 & 0 &  &  \\ 
\vdots &  & \ddots &  \\ 
0 & 0 & \cdots & 0%
\end{array}%
\right] ,\text{ \ \ }\eta =\left[ 
\begin{array}{c}
\eta_{p+1} \\ 
\eta_{p+2} \\ 
\vdots \\ 
\eta_{T}%
\end{array}%
\right] ,\text{ \ \ }Q=\left[ 
\begin{array}{cccc}
Q_{p+1} &  &  & 0 \\ 
& Q_{p+2} &  &  \\ 
&  & \ddots &  \\ 
0 &  &  & Q_{T}%
\end{array}%
\right] .
\end{eqnarray*}

Unlike a more general state-space model where equation (\ref{model2}) has a transition matrix that includes unknown parameters to be estimated, the matrix formulation of the TV parameter model is largely simplified. For example, matrix $C$ is often called the random walk generating matrix (e.g., \citet{tanaka2017tsa}), which is non-singular, and there are no free parameters to be estimated in the matrix. In addition, if $H_{t}$ and $Q_{t}$ are time invariant; that is, if there are no GARCH effects or stochastic volatility in the model, the matrices $H$ and $Q$ are simplified substantially.

For simplicity, we assume $b_{0}$ is known and non-stochastic; hence, $P_{0}=0$.\footnote{This assumption does not change our conclusions below. The main difference is that $Var\left( \beta \right) =C\left( P_{0}^{\ast }+Q\right) C^{\prime }$ and $Var\left( Y_{T}\right) =ZC\left( P_{0}^{\ast }+Q\right) C^{\prime}Z^{\prime }+H=\Omega $. An exception is when the diffuse prior is used, and the likelihood function is computed excluding the first few observations. In such a case, the estimates of the unknown intercept parameters would be different across the two approaches.}

\subsection{Model with Time-Invariant Intercepts}

While our model, (\ref{model1}) and (\ref{model2}), and its matrix formulation, (\ref{mat1}) and (\ref{mat2}), are flexible enough to admit time-varying coefficients, it is sometimes assumed that the class of TV-AR models has time-invariant intercepts. For the purpose of deriving the likelihood function, here, we modify our model to admit time-invariant intercepts. Suppose we have a $k\times k$ vector of time-invariant intercepts, $v$, in our model. Then, (\ref{model1}) and (\ref{model2}) become 
\begin{eqnarray}
y_{t} &=&v+Z_{t}\beta_{t}+\varepsilon_{t} \\
\beta_{t} &=&\beta_{t-1}+\eta_{t}.
\end{eqnarray}%
In this case, it is convenient to use a matrix form to derive the likelihood function. With the vector of intercepts, our model in matrix form, (\ref{mat1}) and (\ref{mat2}), is then modified to 
\begin{eqnarray}
Y_{T} &=&\mathcal{I}v+Z\beta +\varepsilon  \label{m_int} \\
\beta &=&C\left( b_{0}^{\ast }+\eta \right) ,  \label{m1_int}
\end{eqnarray}%
where $\mathcal{I=}\left[ 
\begin{array}{cccc}
I_{k} & I_{k} & \cdots & I_{k}%
\end{array}%
\right] ^{\prime },$ and\ $I_{k}$ is a $k\times k$ identity matrix. Similar to our assumption that time-varying intercepts, if they exist, are unknown, we assume that the vector of time-invariant intercepts, $v$, is the unknown parameter vector.

\subsection{The Likelihood Function}

Since we have the assumption that $\varepsilon $ and $\eta $ are normally distributed, our matrix formulation of (\ref{m_int}) and (\ref{m1_int}) allows us to write the log likelihood function for $Y_{T}$ given the covariance matrices of the errors ($H$ and $Q$), intercepts ($v$), and the initial value vector ($b_{0}^{\ast }$) as: 
\begin{equation}
\begin{split}
\log p\left( Y_{T}|H,Q,v,b_{0}^{\ast }\right)&=-\frac{\left( T-p\right) k}{2}%
\log 2\pi -\frac{1}{2}\log \left\vert \Omega \right\vert \\
&\quad -\frac{1}{2}\left(Y_{T}-ZCb_{0}^{\ast }-\mathcal{I}v\right) ^{\prime }\Omega ^{-1}\left(
Y_{T}-ZCb_{0}^{\ast }-\mathcal{I}v\right),  \label{lik_f}
\end{split}
\end{equation}
where 
\begin{equation*}
\Omega =H+ZCQC^{\prime }Z^{\prime }.
\end{equation*}%
The likelihood function is further simplified when we have no time-invariant intercepts. This is true even if we have time-varying intercepts because the time-varying intercepts are included in vector $\beta$. As a result, in such a case, our log likelihood function becomes 
\begin{equation}
\begin{split}
\log p\left( Y_{T}|H,Q,b_{0}^{\ast }\right) &=-\frac{\left( T-p\right) k}{2}\log 2\pi \\
&\quad -\frac{1}{2}\log \left\vert \Omega \right\vert -\frac{1}{2}\left(Y_{T}-ZCb_{0}^{\ast }\right)^{\prime }\Omega^{-1}\left(Y_{T}-ZCb_{0}^{\ast }\right).  \label{lik_s}
\end{split}
\end{equation}

Interestingly, provided that $H,$ $Q,$ and $b_{0}^{\ast}$ are known, the likelihood function does not involve the parameter vector of our main interest, $\beta$.

%% file: alt_tvp_est_tvar.tex
\section{Estimation of the Time-Varying AR Models}\label{alt_tvp_sec3}

\subsection{Regression Lemma and Kalman Smoothing}
Before showing the equivalence of our estimator and the Kalman smoother, let us clarify what the Kalman smoother does when the model is described by equations (\ref{model1}) and (\ref{model2}). According to \citet{durbin2012tsa}, the Kalman-smoothed state of $\beta$ is given by the expectation of $\beta$ conditional on the information pertaining to all observations of $y_{t}$:
\begin{equation}
\widetilde{\beta }=E\left[ \beta |Y_{T}\right] =E\left[ \beta \right]
+Cov\left( \beta ,Y_{T}\right) Var\left( Y_{T}\right) ^{-1}\left( Y_{T}-E%
\left[ Y_{T}\right] \right) .  \label{Kal1}
\end{equation}%
To derive equation (\ref{Kal1}), note that we assume normal errors. The variance of $\beta $, given all the observations $Y_{T}$, is
\begin{equation}
Var\left( \beta |Y_{T}\right) =Var\left( \beta \right) -Cov\left( \beta
,Y_{T}\right) Var\left( Y_{T}\right) ^{-1}Cov\left( \beta ,Y_{T}\right)
^{\prime }.  \label{Kal2}
\end{equation}%
Note that the Kalman-smoothed estimate and its mean squared error (MSE) are given by (\ref{Kal1}) and (\ref{Kal2}), respectively. \citet{durbin2012tsa} call these equations the regression lemma, which derives the mean and variance of the distribution of $\beta $ conditional on $Y_{T}$, assuming the joint distribution of $\beta $ and $Y_{T}$ is a multivariate normal distribution. It follows that for (\ref{mat1}) and (\ref{mat2}), the Kalman-smoothed estimate is 
\begin{equation}
\widetilde{\beta }=E\left[ \beta |Y_{T}\right] =Cb_{0}^{\ast }+CQC^{\prime
}Z^{\prime }\Omega ^{-1}\left( Y_{T}-ZCb_{0}^{\ast }\right) ,
\label{smooth1}
\end{equation}%
and the conditional variance (or MSE) of the smoothed estimate is 
\begin{equation}
Var\left( \beta |Y_{T}\right) =CQC^{\prime }-CQC^{\prime }Z^{\prime }\Omega
^{-1}ZCQC^{\prime },  \label{smoothP}
\end{equation}%
where $\Omega =H+ZCQC^{\prime }Z^{\prime }$. Equations (\ref{smooth1}) and (\ref{smoothP}) are obtained by utilizing the fact that $Cov\left(\beta, Y_{T}\right) =Var\left( \beta \right) Z^{\prime }=CQC^{\prime }Z^{\prime }$ and $Var\left( Y_{T}\right) =ZCQC^{\prime }Z^{\prime }+H=\Omega $, and by substituting them into equations (\ref{Kal1}) and (\ref{Kal2}). It is well known that (\ref{smooth1}) is a minimum-variance linear unbiased estimate of $\beta$, given $Y_{T}$, even though we do not assume the errors are normally distributed (see \citet{durbin2012tsa}, among others).

\subsection{The Equivalence of the GLS-Based Estimator and the Kalman Smoother}

It is possible to write equations (\ref{mat1}) and (\ref{mat2}) in another matrix form to apply a conventional regression analysis:
\begin{equation}
\left[ 
\begin{array}{c}
Y_{T} \\ 
-b_{0}^{\ast }%
\end{array}%
\right] =\left[ 
\begin{array}{c}
Z \\ 
-C^{-1}%
\end{array}%
\right] \beta +\left[ 
\begin{array}{c}
\varepsilon \\ 
\eta%
\end{array}%
\right].  \label{reg}
\end{equation}%
This specification is similar to that of \citet{duncan1972ldr} and of \citet{maddala1998urc}. The main difference between our specification and that of \citet{duncan1972ldr} is that the former applies to a time-varying parameter model, while the latter is for a more general state-space model, which allows the transition equation to have a transition matrix $F$ (i.e., when equation (\ref{model2}) is $\beta_{t}=F\beta_{t-1}+\eta_{t}$). Since we do not need to estimate the transition matrix, our regressors in equation (\ref{reg}) are all known. In contrast, \citet{duncan1972ldr} assume the matrix $F$ is known, which renders their estimation impractical. The original form of \citet[pp.469--470]{maddala1998urc} is similar to ours, but it is a general form for a scalar $y_{t}$. Hence, seemingly, it does not aim to deal with the autoregressive part of the time-varying parameter models nor does it consider vector processes. For a simple scalar case, however, \citet{maddala1998urc} point out that GLS is equivalent to the Kalman-smoothed estimate without formal proof.

As mentioned by \citet{duncan1972ldr}, the confusion concerning the similarities and differences between Kalman filtering and the conventional regression model stems from the fact that the former is the expectation of $\beta$, conditional on the information about $Y_{T}$, which is the linear projection of $\beta $ onto the space spanned by $Y_{T}$ (provided that the errors are normally distributed); while the latter is a linear projection of the dependent variable onto the space spanned by the regressor, which is the projection of the left hand side on equation (\ref{reg}) onto the space spanned by $\left[ 
\begin{array}{cc}
Z^{\prime} & -C^{-1\prime}%
\end{array}%
\right]^{\prime}$. However, \citet{duncan1972ldr} essentially show that GLS for (\ref{reg}) up through the time-$t$ observation yields the Kalman filtered estimate. Therefore, a natural conjecture is that we obtain the Kalman-smoothed estimate of $\beta$ when GLS is applied to all the observations, $Y_{T}$. In fact, this conjecture is correct, and we have the following Proposition.

\begin{proposition}
The GLS estimator of model (\ref{reg}) yields the Kalman-smoothed estimates (\ref{smooth1}) and its mean squared error matrix (\ref{smoothP}).
\end{proposition}

\begin{proof}
See the Appendix.
\end{proof}

\subsection{The GLS Estimator Under the Presence of Time-Invariant Intercepts}

As we discuss in the previous section, our model admits time-invariant intercepts. Therefore, it is straightforward to define the GLS estimator for such models. To do so, assuming that the time-invariant intercepts are unknown, let us define the vector of unknown parameters, $\beta^{\ast}=
\left[ 
\begin{array}{cc}
v^{\prime } & \beta ^{\prime }%
\end{array}%
\right]^{\prime}$. Then, the matrix form for regression that is analogous
to (\ref{reg}) is 
\begin{equation}
\left[ 
\begin{array}{c}
Y_{T} \\ 
-b_{0}^{\ast }%
\end{array}%
\right] =\left[ 
\begin{array}{cc}
\mathcal{I} & Z \\ 
0 & -C^{-1}%
\end{array}%
\right] \beta^{\ast }+\left[ 
\begin{array}{c}
\varepsilon \\ 
\eta%
\end{array}%
\right].  \label{m2_int}
\end{equation}

Here, one of the advantages of utilizing the regression approach (\ref{m2_int}) over Kalman smoothing (\ref{m1_int}) is that the unknown intercept vector $v$ is estimated simultaneously with $\beta $. Then, it can be shown that the GLS estimate $\widehat{v}$ is indeed the maximum likelihood estimate.

\begin{proposition}
The GLS estimate $\widehat{v}$ of model (\ref{m2_int}) is the maximum likelihood estimate (MLE) of (\ref{m1_int}), $\widehat{v}_{ML}$ conditional on $H,Q,$ and $b_{0}^{\ast }$.
\end{proposition}

\begin{proof}
From the likelihood function, (\ref{lik_f}), the normal equations pertaining to $v$ are 
\begin{equation*}
\mathcal{I}^{\prime}\Omega^{-1}\left(Y_{T}-ZCb_{0}^{\ast}-\mathcal{I}\widehat{v}_{ML}\right)=0.
\end{equation*}%
Therefore, the MLE for $v$ is%
\begin{equation}
\widehat{v}_{ML}=\left(\mathcal{I}^{\prime}\Omega^{-1}\mathcal{I}\right)^{-1}\mathcal{I}^{\prime}\Omega^{-1}\left(Y_{T}-ZCb_{0}^{\ast}\right).
\label{ml_v}
\end{equation}

Now, the GLS estimates for $\beta ^{\ast }$ in model (\ref{m2_int}) are 
\begin{equation}
\begin{split}
\widehat{\beta }^{\ast }=\left[ 
\begin{array}{c}
\widehat{v} \\ 
\widehat{\beta }%
\end{array}%
\right]&=\left[ 
\begin{array}{cc}
\mathcal{I}^{\prime }H^{-1}\mathcal{I} & \mathcal{I}^{\prime }H^{-1}Z \\ 
Z^{\prime }H^{-1}\mathcal{I} & Z^{\prime }H^{-1}Z+C^{\prime -1}Q^{-1}C^{-1}%
\end{array}%
\right] ^{-1}\\
&\quad \quad \quad \quad \quad \quad \quad \quad \left[ \left( 
\begin{array}{c}
\mathcal{I}^{\prime }H^{-1} \\ 
Z^{\prime }H^{-1}%
\end{array}%
\right) Y_{T}+\left( 
\begin{array}{c}
O \\ 
C^{-1\prime }Q^{-1}%
\end{array}%
\right) b_{0}^{\ast }\right].  \label{GLS}
\end{split}
\end{equation}

Using the Lemma, we arrive at the following (see the appendix for details):%
\begin{equation*}
\widehat{v}=\left( \mathcal{I}^{\prime }\Omega ^{-1}\mathcal{I}\right) ^{-1}%
\mathcal{I}^{\prime }\Omega ^{-1}\left( Y_{T}-ZCb_{0}^{\ast }\right) .
\end{equation*}%
This proves $\widehat{v}_{ML}=\widehat{v}$.
\end{proof}

\begin{proposition}
The GLS estimate $\widehat{\beta }$ of model (\ref{m2_int}) is the Kalman-smoothed estimate of model (\ref{m1_int}).
\end{proposition}

\begin{proof}
Thanks to the intercept, the Kalman-smoothed estimate is now 
\begin{equation}
\widetilde{\beta }=Cb_{0}^{\ast }+CQC^{\prime }Z^{\prime }\Omega ^{-1}\left(
Y_{T}-\mathcal{I}v-ZCb_{0}^{\ast }\right) .
\end{equation}%
From (\ref{GLS}), it follows that 
\begin{equation*}
\widehat{\beta }=Cb_{0}^{\ast }+CQC^{\prime }Z^{\prime }\Omega ^{-1}\left(
Y_{T}-\mathcal{I}\widehat{v}-ZCb_{0}^{\ast }\right) .
\end{equation*}%
We prove the equivalence.
\end{proof}

It is clear that the GLS-based approach can compute the Kalman-smoothed $\beta$ and estimate the unknown intercepts, $v$, simultaneously. The next question is how we can obtain the statistical inference about $\widehat{\beta}$. More precisely, at issue is whether the GLS-based approach yields the same MSE as the Kalman smoother. The answer to this question is negative for $\widehat{\beta}$.

\begin{proposition}
The mean squared error of the Kalman smoothed estimate is 
\begin{equation}
Var\left( \beta |Y_{T}\right) =CQC^{\prime }-CQC^{\prime }Z^{\prime }\Omega
^{-1}ZCQC^{\prime },  \label{Varb}
\end{equation}%
whereas the variance estimated from the GLS-based approach (\ref{GLS}) is 
\begin{equation}
Var\left( \widehat{\beta }\right) =CQC^{\prime }-CQC^{\prime }Z^{\prime
}\Omega ^{-1}ZCQC+CQC^{\prime }Z^{\prime }\Omega ^{-1}\mathcal{I}\left( 
\mathcal{I}^{\prime }\Omega ^{-1}\mathcal{I}\right) ^{-1}\mathcal{I}^{\prime
}\Omega ^{-1}ZCQC^{\prime }.  \label{Varbhat}
\end{equation}
\end{proposition}

\begin{proof}
See the Appendix.
\end{proof}

The difference between the Kalman-smoothed $Var\left(\beta|Y_{T}\right)$ and the GLS-based variance $Var\left(\widehat{\beta}\right)$ is $CQC^{\prime }Z^{\prime }\Omega ^{-1}\mathcal{I}\left( \mathcal{I}^{\prime}\Omega ^{-1}\mathcal{I}\right) ^{-1}\mathcal{I}^{\prime }\Omega^{-1}ZCQC^{\prime }$, which pertains to the estimation of $v$. If we did not have to estimate $v$, (as we assume for the Kalman-smoothed estimate)$ Var\left( \beta |Y_{T}\right) $ and $Var\left( \widehat{\beta }\right)$ would be the same. In other words, if $v$ is known, the MSE of the Kalman-smoothed estimate is the same as the variance of the GLS-based estimate. As a matter of fact, if $Var\left( \widehat{v}\right) =\left(\mathcal{I}^{\prime }\Omega ^{-1}\mathcal{I}\right) ^{-1}=0$, the two estimates would be identical. This result reflects that the two approaches yield the same estimate and MSE, as in Proposition 1. Nevertheless, what is important here is that we can obtain (\ref{Varb}) by utilizing the estimated variance of $\widehat{\beta }^{\ast }$ of (\ref{GLS}). More specifically, we can estimate the MSE of the Kalman-smoothed estimate by 
\begin{equation*}
Var\left( \beta |Y_{T}\right) =Var\left( \widehat{\beta }\right) -Cov\left( 
\widehat{\beta },\widehat{v}\right) Var\left( \widehat{v}\right)
^{-1}Cov\left( \widehat{\beta },\widehat{v}\right) ^{\prime }.
\end{equation*}

\subsection{GLS in Practice}

As we have seen in previous subsections, under the condition that the variance-covariance matrices of errors ($H$ and $Q$) are known, the GLS estimator of $\beta $ is identical to the Kalman smoothed estimates. However, in practice, those variance-covariance matrices are generally unknown. To find the FGLS estimator, it is often used the two-step approach: First, one can estimate $\beta $ by ordinary least squares (OLS). Then, the OLS residuals are used to compute the estimates of $H$ and $Q$, denoted as $\widehat{H}$ and $\widehat{Q}$, respectively. As the second step, FGLS is applied to our model assuming $\widehat{H}$ and $\widehat{Q}$ are the variance-covariance matrices of $\varepsilon$ and $\eta$, respectively.

However, there are two problems pertaining to FGLS. First, $H$ and $Q$ may involve too many unknown parameters. For example, when we deal with a TV-VAR(p) model with $k$ variables, $H$ has $\left(T-p\right) $ of $k\times k$ matrices, and $Q$ has the same number of $k\left(kp+1\right) \times k\left(kp+1\right) $ matrices. The second problem is possible heteroskedasticity. Suppose that $\varepsilon $ is much greater in magnitude than $\eta $. More precisely, when the average trace of $H$ is much larger than the average trace of $Q$. Then, in such a case, our GLS-based approach has heteroskedasticity in regression equation (\ref{reg}), potentially causing imprecise estimation of $\beta $. This concern is largely mitigated when the average trace of $H$ and the average trace of $Q$ are a similar size.

As a solution to these two problems, we propose the following FGLS procedure.

\begin{itemize}
\item Step 1. We estimate model (\ref{reg}) by OLS, and obtain the estimate of $\beta$ by OLS, $\widehat{\beta}^{O}$. From the OLS residuals, $\widehat{\varepsilon}_{t}$ and $\widehat{\eta }_{t}$, we construct the first step estimates of $H_{t}$ and $Q_{t}$:
\begin{equation*}
\widehat{H}_{t}=\frac{1}{T-p}\sum_{t=p+1}^{T}\widehat{\varepsilon }_{t}%
\widehat{\varepsilon }_{t}^{\prime }\text{ \ and }\widehat{Q}_{t}=\frac{1}{%
T-p}\sum_{t=p+1}^{T}\widehat{\eta }_{t}\widehat{\eta }_{t}^{\prime }.
\end{equation*}%
Then, to construct the estimates of $H$ and $Q$, denoted $\widehat{H}^{O}$ and $\widehat{Q}^{O}$, respectively, we set $\widehat{H}_{p+1}=\widehat{H}_{p+2}=\cdots=\widehat{H}_{T}$ and $\widehat{Q}_{p+1}=\widehat{Q}_{p+2}=\cdots=\widehat{Q}_{T}$. This is to assume that the variances of $\varepsilon$ and $\eta$ are time-invariant. This assumption is by no means desirable because a number of studies pertaining to the class of TV-VAR models have focused on the stochastic volatility models, which requires $\widehat{Q}_{t}\neq \widehat{Q}_{t+1}$, for example. The simulations in the next section will reveal how severely this assumption affects our estimation when stochastic volatility is present. With $\widehat{H}^{O}$ and $\widehat{Q}^{O}$, the log likelihood is computed by (\ref{lik_f}) or (\ref{lik_s}).

\item Step 2 (1FGLS). Given $\widehat{H}^{O}$ and $\widehat{Q}^{O}$, we apply FGLS to obtain $\widehat{\beta}^{G1}$, which is the FGLS or 1FGLS estimate of $\beta$. We also compute the estimates of $H$ and $Q$, denoted as $\widehat{H}^{G1}$ and $\widehat{Q}^{G1}$, respectively, in the same way as we computed $\widehat{H}^{O}$ and $\widehat{Q}^{O}$ in the first step. Then, the value of the log likelihood function is computed.

\item Step 3 (2FGLS). We repeat Step 2, computing $\widehat{\beta}^{G2}$, which is the (second-time) FGLS or 2FGLS of $\beta$. Then, the value of the log likelihood function is computed.
\end{itemize}

In summary, our procedure is based on the assumptions that the error terms have time-invariant variances and that heteroskedasticity arising from different sizes of $H$ and $Q$ can be correctly handled by repeated use of FGLS. To validate our assumptions and procedure, we investigate the degree to which our procedure precisely estimate the true $\beta$ via simulations in the next section.

%% file: alt_tvp_simulation.tex
\section{Simulations}\label{alt_tvp_sec4}

Among some influential empirical studies in the literature of TV-VAR, both \citet{cogley2001epw,cogley2005dvm} and \citet{primiceri2005tvs} employ a three-variable TV-VAR(2) model. Hence, in our simulation study, we employ the same specification and use simulations to assess how well the GLS-based approach recovers the true time-varying parameters. First, we compute the means and variances of the estimated time-varying parameters, and compare them with the means and variances of the true time-varying parameters. This is to evaluate the GLS-based approach in terms of its accuracy in estimating the time-varying parameter. While comparing the first and second moments of the estimates to those of the true process may not be adequate to determine whether the GLS-based approach yields precise estimates, it is a useful way to grasp the overall accuracy of the estimates.\footnote{In addition, we can compute the values of the log likelihood function to evaluate whether the repeated use of FGLS improves accuracy in estimation. A general tendency throughout our simulation is that aforementioned 2FGLS has a higher likelihood value than 1FGLS.}

Second, we consider the possibility of the pile-up problem. According to \citet{primiceri2005tvs}, the Bayesian approach is preferred when estimating time-varying parameter models. This is because, among other reasons, the Bayesian approach can potentially avoid the pile-up problem. It is not immediately obvious to what extent the problem affects our estimate, because the literature (e.g., \citet{shephard1990ope}) provides theoretical explanations only for limited (simple) cases. On the other hand, our model can have a vector of time-varying terms ($\beta_{t}$), unlike prior studies that analyzed scalar time-varying terms for simplicity. Therefore, it is reasonable to conduct a simulation study to reveal the extent to which our GLS-based approach suffers from the pile-up problem. Because the concern over the pile-up problem becomes stronger when the variance of the state-equation error is small, or when SNR is small, we study the performance of the GLS-based approach more comprehensively by altering SNR in the data generating process.

Third, we also evaluate the performance of the GLS-based approach when stochastic volatility and non-Gaussian errors are present. The reason for investigating the effect of stochastic volatility on the GLS-based approach is that macroeconomic research, including \citet{cogley2005dvm} and \citet{primiceri2005tvs}, has been allowing such shocks in the TV-VAR model. While the GLS-based approach does not require the assumption of i.i.d. errors to obtain the estimate of $\beta_{t}$, we are interested in the extent to which the accuracy of the GLS-based approach is affected by the stochastic volatility of the errors. For the non-Gaussian errors, our focus is possible structural breaks in the time-varying coefficients, $\beta_{t}$. By allowing a mixture of normal errors, as explained in the following subsection, we can model structural breaks or abrupt changes in $\beta_{t}$, as opposed to gradual changes that the time-varying model generally assumes. Our simulation study is expected to shed light on the performance of the GLS-based approach when such errors are present.

Finally, as we mention in Section \ref{alt_tvp_sec1}, we consider OLS as a component of the GLS-based approach, and hence, we study the performance of OLS using simulations. This is because, generally speaking, the performance FGLS relative to OLS is not clear especially when we have a small sample.

\subsection{The Data Generating Process}\label{alt_tvp_sec4_1}

We generate pseudo data by the system of equations (\ref{mat1}) and (\ref{mat2}) with $T=\left\{ 100,250\right\}$, $H=\left\{0.02^{2}I,0.2^{2}I,1^{2}I,10^{2}I \right\}$, and $Q=\left\{0.03^{2}I\right\} $. By changing the variance of the error to the observation equation, we consider the role of SNR. In what follows, we define the SNR as the average trace of the variance-covariance matrix of $\eta_{t}$ relative to the average trace of the variance-covariance matrix of $\varepsilon_{t}$: In our simulation, we consider SNRs for $0.03^{2}/0.02^{2}$, $0.03^{2}/0.2^{2}$, $0.03^{2}/1^{2}$, and $0.03^{2}/10^{2}$. The SNR is particularly important when we consider the possibility of the pile-up problem, which will be discussed in the next section. For the initial values, we set $b_{0}^{\ast }=0$.

\subsubsection{Non-Gaussian Errors}

The original motivation to employ time-varying models for macroeconomic research was to allow for gradual change in $\beta_{t}$. However, it is possible that there are some structural breaks or abrupt changes in $\beta_{t}$, which means that $\beta_{t}$ is almost constant over time until some point in the sample, for example, $T_{b}$; it then jumps to a different level afterward. One way to model such a break is to assume non-Gaussian errors for $\eta_{t}$. In particular, we assume mixtures of normal distributions (among others, \citet{perron2009ltb}) for each element of error vector $\eta_{t}$:
\begin{equation*}
\eta_{it}=\lambda_{t}\zeta_{1,t}+\left( 1-\lambda_{t}\right) \zeta_{2,t}
\end{equation*}%
where 
\begin{eqnarray*}
\lambda_{t} &\sim &i.i.d.Bernoulli\left( 0.95\right) \\
\zeta_{1,t} &\sim &N\left( 0,0.03^{2}\right) ,\text{ \ }\zeta_{2,t}\sim
N\left( 0,0.1^{2}\right) .
\end{eqnarray*}

Intuitively, with a probability of 95\%, $\eta_{t}$ is $\zeta_{1,t}$, which is drawn from a normal distribution with a small variance. This small $\eta_{t}$ keeps $\beta_{t}$ nearly constant over time. However, a large $\eta_{t}$, which is $\zeta_{2,t}$, is drawn from a normal distribution with a (relatively) large variance. This $\eta_{t}$ causes $\beta_{t}$ to jump to a new level, with 5\% probability. Since we use the assumption of Gaussian error to derive the equivalence between GLS and the Kalman smoothed estimator, the effect of non-Gaussian errors on the accuracy of the GLS estimator in estimating $\beta_{t}$\ should be evaluated via simulations.

\subsubsection{Stochastic Volatility and Autoregressive Stochastic Volatility}

As \citet{cogley2005dvm} argue, in response to the criticisms of \citet{cogley2001epw}, it is more flexible and realistic to assume that the variance of the shock $\varepsilon_{t}$\ is time varying. Intuitively, not all shocks are generated from the same i.i.d. process. One peculiar feature of the GLS-based approach is that it can handle the heteroskedasticity in $H_{t}$ and $Q_{t}$. This means that, at least theoretically, we can estimate the time-varying model with stochastic volatility, such as the one used by \citet{primiceri2005tvs}. It is also possible that the error term $\varepsilon_{t}$ follow the autoregressive stochastic volatility process described by \citet{taylor2007mft} and elsewhere.

However, in general, FGLS is merely a remedy to more precisely estimate the coefficients (in our case, $\beta_{t}$) when heteroskedasticity is present, and FGLS is not primarily designed to estimate the process that the error term (or its variance) follows.

Nevertheless, we use the following data generating process to assess the performance of the GLS-based approach.
\begin{eqnarray*}
\varepsilon_{it} &=&\sqrt{h_{i,t}}\xi_{t} \\
\log h_{i,t} &=&\rho \log h_{i,t-1}+e_{t}
\end{eqnarray*}%
where $\rho =1$ when stochastic volatility is considered, and $\rho=0.9$ when autoregressive stochastic volatility is considered; and $\varepsilon_{it}$ is the $i$-th element of $\varepsilon_{t}$. We assume $\log h_{i,0}=0 $, $e_{t}\sim \mathcal{N}\left( 0,0.02^{2}\right) $, and $\xi_{t}\sim \mathcal{N}\left( 0,1\right) $.

\subsection{The Mean and the Variance of Estimated $\protect\beta_{t}$, and the Likelihood}

Since our simulation is of a TV-VAR(2) model with time-varying intercepts, $\beta_{t}$ is a 21 $\times $ 1 vector. Let $\beta_{t,i,n}$ denote the true (DGP) $\beta_{t,i,n}$ (i.e., the $i$-th element of vector $\beta_{t,n}$)and let $\widehat{\beta}_{t,i,n}^{G}$ denote the GLS-based estimate for $\beta_{t,i,n}$. \ Since, in practice, we do not know $b_{0}^{\ast }$ when estimating $\beta_{t,i,n}$, we estimate $b_{0}^{\ast }$ as the coefficients vector from a full-sample time-invariant (usual) VAR(2) model before estimating $\beta_{t,i,n}$ by GLS. The sample means and the sample standard deviations of the estimate over the sample period are then computed: 
\begin{eqnarray}
\overline{\widehat{\beta }}_{i,n}^{G} &=&\frac{1}{T-p}\sum_{t=p+1}^{T}%
\widehat{\beta }_{t,i,n}^{G}  \label{mean_in} \\
sd\left( \widehat{\beta }_{i,n}^{G}\right) &=&\sqrt{\frac{1}{T-p-1}%
\sum_{t=p+1}^{T}\left( \widehat{\beta }_{t,i,n}^{G}-\overline{\widehat{\beta 
}}_{i,n}^{G}\right) ^{2}}.  \label{vol_in}
\end{eqnarray}

Similarly, we compute those of the true (data generating) process:
\begin{eqnarray}
\overline{\beta }_{i,n} &=&\frac{1}{T-p}\sum_{t=p+1}^{T}\beta_{t,i,n}
\label{mean_DGP} \\
sd\left( \beta_{i,n}\right) &=&\sqrt{\frac{1}{T-p-1}\sum_{t=p+1}^{T}\left(
\beta_{t,i,n}-\overline{\beta }_{i,n}\right) ^{2}}.  \label{vol_DGP}
\end{eqnarray}

By (\ref{mean_in}) and (\ref{vol_in}) and there DGP counterparts, (\ref{mean_DGP}) and (\ref{vol_DGP}), we have 21 means and standard deviations for each replication. After $N=1,000$ replications, we compute the averages of $\overline{\widehat{\beta }}_{i,n}^{G}$, $\overline{\beta }_{i,n}$, $sd\left( \widehat{\beta }_{i,n}^{G}\right) $, and $sd\left( \beta_{i,n}\right) $ over the replications. We then have 21 means of time-varying parameters and 21 means of standard deviations (i.e., $i=1,2,\ldots ,21$).
\begin{eqnarray}
m_{i}^{G} &=&\frac{1}{N}\sum_{n=1}^{N}\overline{\widehat{\beta }}_{i,n}^{G};%
\text{ \ }m_{i}=\frac{1}{N}\sum_{n=1}^{N}\overline{\beta }_{i,n}
\label{av_mean} \\
s_{i}^{G} &=&\frac{1}{N}\sum_{n=1}^{N}sd\left( \widehat{\beta }%
_{i,n}^{G}\right) ;\text{ }s_{i}=\frac{1}{N}\sum_{n=1}^{N}sd\left( \beta
_{i,n}\right)  \label{av_vol}
\end{eqnarray}%
Since both $m$ and $s$ are aggregate means, a small difference between $m$ and $m^{G}$ or between $s$ and $s^{G}$\ is only an indication that the GLS-based approach works well. Hence, we further investigate the similarities of $\beta $ and $\widehat{\beta }^{G}$. Comparing each element of $\beta $, we define the distance, \textquotedblleft $dist,$ \textquotedblright\ as follows.
\begin{equation}
dist_{i}=\frac{1}{N\left( T-p\right) }\sum_{n=1}^{N}\sum_{t=p+1}^{T}\left%
\vert \beta_{t,i,n}-\widehat{\beta }_{t,i,n}^{G}\right\vert .  \label{dist}
\end{equation}

Similarly, we compare the standard deviations of each element of $\beta $ as a ratio of the standard deviation of $\widehat{\beta }_{i,n}^{G}$ to the standard deviation of the true process, $\beta_{i,n}$:
\begin{equation}
rat_{i}=\frac{1}{N}\sum_{n=1}^{N}\frac{sd\left( \widehat{\beta}_{i,n}^{G}\right) }{sd\left( \beta_{i,n}\right) }.  \label{rat}
\end{equation}

In this simulation study, we focus on both $dist_{i}$ and $rat_{i}$. Our criteria for a good estimator are whether $dist_{i}$ of an estimate is close to zero and whether $rat_{i}$ of that estimate is close to one.

\subsection{Simulation Results 1: The Signal-to-noise Ratio, Sample Size, and the Precision of Estimation}
\begin{center}
(Tables \ref{alt_table1} and \ref{alt_table2} around here)
\end{center}

Tables \ref{alt_table1} and \ref{alt_table2} display the medians of $dist_{i}$ and $rat_{i}$ as well as the medians of $m_{i}$ and $s_{i}$ for $T=100$ and $T=250$, respectively. OLS works relatively well when the SNR is relatively large because, as Tables \ref{alt_table1} and \ref{alt_table2} show, the median distance of the estimate from the true process (i.e., $dist_{i}$) is small, and the median sample variance of the estimated $\beta_{t}$ is close to that of the true process (i.e., $rat_{i}$ is close to one).\footnote{Throughout this simulation study, we use bold numbers to highlight the best (the smallest median $dist$ and the median $rat$ closest to one) estimation method of the three (OLS, 1FGLS, and 2FGLS).} On the contrary, 2FGLS works relatively well when the SNR is small. General tendencies from Table \ref{alt_table1} can be summarized as follows. First, OLS and 1FGLS share largely the same characteristics. However, estimated $\beta_{t}$ by 2FGLS is much smaller in magnitude than those estimated by OLS and 1FGLS. Second, OLS, 1FGLS, and 2FGLS all tend to have larger $rat_{i}$ as SNR increases. More precisely, OLS, 1FGLS, and 2FGLS overestimate the volatility of $\beta_{t}\ $when SNR is very small, and underestimate when SNR is very large. Third, as for the median distance of the estimate from the true process (i.e., $dist_{i}$), the best case is when SNR is 2.25. This phenomenon is easy to understand because both too small and too large SNR make the estimation of $\beta_{t}$ difficult since SNR far from one means the degree of heterogeneity is quite serious. In such a situation, it is easy to imagine that OLS does not do a good job in recovering $\beta_{t}$; and 1FGLS is probably not a good way to implement FGLS.

What is the effect of increasing the sample size? A comparison of Tables \ref{alt_table1} and \ref{alt_table2} shows that the degrees of overestimating or underestimating the volatility of $\beta_{t}$ are largely mitigated for OLS and 1FGLS when the sample size increases from 100 to 250. At the same time, the median distances of the estimate from the true process for OLS and 1FGLS become smaller as the sample size increases. It goes to to show that the accuracy of OLS and 1FGLS improves with the sample size. Notably, however, such effects of increased sample size do not clearly hold for 2FGLS.

\subsection{Simulation Results 2: The Effects of Non-i.i.d. and Non-Gaussian Errors}
\begin{center}
(Table \ref{alt_table3} around here)
\end{center}

Table \ref{alt_table3} demonstrates the effect of non-Gaussian errors as well as stochastic volatility and stochastic autoregressive errors. The general tendencies that appear in the Gaussian error case (Tables \ref{alt_table1} and \ref{alt_table2}) are preserved: both OLS and 1FGLS overestimate the volatility of $\beta_{t}$; yet the degree of overestimation is largely mitigated when the sample size increases; 2FGLS underestimates the volatility of $\beta_{t}$, and increasing the sample size does not help 2FGLS improve the accuracy in the estimation of $\beta_{t}$. Remarkably, given the value of the autoregressive parameter $\rho$, there is negligible difference between the stochastic volatility and autoregressive stochastic volatility cases.
\begin{center}
(Tables \ref{alt_table4} and \ref{alt_table5} around here)
\end{center}

When only the non-Gaussian error is considered, as Tables \ref{alt_table4} and \ref{alt_table5} show, we obtain mostly the same results as those presented in Tables \ref{alt_table1} and \ref{alt_table2}. Once again, SNR is key to determining the estimated sample variance of $\beta_{t} $ relative to its true sample variance. In other words, the degree of overestimation (underestimation) depends on SNR. Similar to the results in Tables \ref{alt_table1} and \ref{alt_table2}, the larger sample size generally helps the estimation by OLS and 1FGLS in that the degree of overestimation or underestimation is largely reduced when the sample size increases. In addition, for OLS and 1FGLS, the median distance between true and estimated $\beta_{t}$ also becomes smaller with the sample size. However, this tendency does not apply to 2FGLS.
\begin{center}
(Table \ref{alt_table6} around here)
\end{center}

What is the effect of scholastic volatility or autoregressive volatility in the observation equation error ($\varepsilon_{t}$) on our estimation? Table \ref{alt_table6} reveals that the results arising from such errors are similar to the small SNR cases in Tables \ref{alt_table1} and \ref{alt_table2}. This is because the observation error ($\varepsilon_{t}$) has a variance larger than one due to the stochastic volatility ($\sqrt{h_{t}}$) term. There is little difference between the results of the stochastic volatility case and the result of the autoregressive stochastic volatility case.

\subsection{Discussion: The Pile-Up Problem}

From our results, the GLS-based approach does not suffer from the pile-up problem. Lower SNRs often lead to overestimation of the volatility of $\beta_{t}$, rather than its underestimation, especially when OLS or 1FGLS is used (Tables \ref{alt_table1} and \ref{alt_table2}). Moreover, the degree of overestimation of the sample variance of $\beta_{t}$ becomes more severe when the sample size is small. This may be puzzling given the fact that OLS and ML are generally equivalent or the fact that GLS and ML are equivalent if the errors are not i.i.d. (i.e., errors having heteroskedasticity or autocorrelation). However, this statement is not true if FGLS fails to deal with non-i.i.d. errors appropriately. It is likely that both OLS and 1FGLS are unable to obtain the estimate of $\beta_{t}$ that is equivalent to ML. This is the reason the GLS-based approach does not suffer from the pile-up problem.

All in all, our simulations seem to suggest that the use of 2FGLS is recommended when the sample size is small; and OLS (and 1FGLS) does a fairly good job in recovering the time-varying parameters when the sample size is large.

%% file: alt_tvp_application.tex
\section{An Application to the TV-VAR(2) with the Interest Rate, Inflation, and Unemployment}\label{alt_tvp_sec5}

\begin{center}
(Figures \ref{alt_fig1} through \ref{alt_fig4} around here)
\end{center}

A number of studies that employ TV-VAR models, including \citet{cogley2005dvm} and \citet{primiceri2005tvs}, focus on recovering the structural parameters from the estimated reduced form. Although our focus is not to identify fundamental shocks or to compute impulse responses, we present the estimated TV-VAR(2) parameter using OLS (Figure \ref{alt_fig1}), FGLSs (Figure \ref{alt_fig2} for 1FGLS and Figure \ref{alt_fig3} for 2FGLS), and the posterior mean of the time-varying approach using the Bayesian MCMC method (Figure \ref{alt_fig4})\footnote{We use the data and MATLAB codes provided by \citet{koop2010bmt}.}.

While the Bayesian MCMC posterior means are virtually time invariant, and the estimates by 2FGLS are slightly more volatile, the estimates by OLS and 1FGLS have much larger volatility.

Interestingly, as detailed in the online appendix, the coefficients on the interest rate vary noticeably over time, and exhibit distinct patterns in the early 1980s (dip), the late 1990s (up), and early 2000s (down). Similar to the Bayesian posterior means, the intercepts (three time-varying coefficients) are largely stable over time.

%% file: alt_tvp_conclusion.tex
\section{Conclusion}\label{alt_tvp_sec6}

The (non-Bayesian) regression-based or GLS-based approach for the time-varying parameter model is presented and assessed from theoretical and simulation aspects. Although this approach has already been (at least partly) used by \citet{ito2014ism,ito2016eme,ito2016tvc}, it is shown that there are, at least, following four advantages to the GLS-based approach. First, this approach does not necessitate Kalman filtering or smoothing, but it does produce equivalent estimates. In addition, this approach is readily applicable to a wide range of time-varying parameter models, such as the TV-AR, TV-VAR, and TV-VEC models, by adjusting the regression matrix accordingly. Second, it is revealed that the GLS-based approach works reasonably well in practice in that it can estimate the time-varying parameters even with non-i.i.d. errors or non-Gaussian errors in the model. This is because GLS can take into account generally heteroskedastic error terms. The ability to deal with non-Gaussian errors is particularly important in empirical studies because it allows us to consider possible abrupt changes in time-varying parameters, instead of gradual changes that are due to Gaussian errors. One caveat is that depending on the sample size and depending on SNR, the most appropriate method, either OLS, 1FGLS, or 2FGLS, should be chosen. More precisely, OLS is acceptable when SNR is not very far from one or when the sample size is not small. However, when SNR is far from one or when the sample size is small, 2FGLS is recommended. The reason why the sample size and SNR are important in choosing one method over the other two methods is that our 1FGLS and 2FGLS are not ideal GLS; hence, they cannot fully take care of heterogeneity arising from our regression equation that includes both the observation equation errors and the state equation errors. However, because we do not maximize unconditional likelihood function with respect to the variances of the errors, and because 1FGLS and 2FGLS are not ideal GLS, the true variances are not precisely estimated, and our GLS-based approach does not suffer from the pile-up problem that often occurs with ML.

%% file: alt_tvp_ack.tex
\section*{Acknowledgments}

We would like to thank James Morley, Daniel Rees, Yunjong Eo, Yohei Yamamoto, Eiji Kurozumi, and conference participants at the 91th Annual Conference of the Western Economic Association International, First International Conference on Econometrics and Statistics, and Macro Reading Group Workshop at the Reserve Bank of Australia for their helpful comments and suggestions. We also acknowledge the financial assistance provided by the Japan Society for the Promotion of Science Grant in Aid for Scientific Research No.26380397 (Mikio Ito), No.15K03542 (Akihiko Noda), No.15H06585 (Tatsuma Wada), Murata Science Foundation Research Grant (Tatsuma Wada), and Okawa Foundation Research Grant (Tatsuma Wada). All data and programs used for this paper are available on request.

%% file: alt_tvp_table.tex
\setcounter{table}{0}
\renewcommand{\thetable}{\arabic{table}}

\clearpage

\begin{table}[tbp]
\caption{Simulation Results $T=100$}\label{alt_table1}
\centering
\normalsize
\begin{tabular}{cclllll}
\hline\hline
$H$ & $Q$ &  & True & OLS & 1FGLS & 2FGLS \\ \hline
$0.002^{2}$ & $0.03^{2}$ & median $m$ & \multicolumn{1}{r}{-0.000} & 
\multicolumn{1}{r}{-0.000} & \multicolumn{1}{r}{-0.001} & \multicolumn{1}{r}{
0.003} \\ 
&  & median $s$ & \multicolumn{1}{r}{0.113} & \multicolumn{1}{r}{0.048} & 
\multicolumn{1}{r}{0.034} & \multicolumn{1}{r}{0.010} \\ 
&  & median $dist$ & \multicolumn{1}{r}{} & \multicolumn{1}{r}{\textbf{0.165}
} & \multicolumn{1}{r}{0.176} & \multicolumn{1}{r}{0.210} \\ 
$SNR=$ & $225$ & median $rat$ & \multicolumn{1}{r}{} & \multicolumn{1}{r}{%
\textbf{0.455}} & \multicolumn{1}{r}{0.325} & \multicolumn{1}{r}{0.097} \\ 
\hline
$0.02^{2}$ & $0.03^{2}$ & median $m$ & \multicolumn{1}{r}{0.000} & 
\multicolumn{1}{r}{0.008} & \multicolumn{1}{r}{0.007} & \multicolumn{1}{r}{
0.003} \\ 
&  & median $s$ & \multicolumn{1}{r}{0.113} & \multicolumn{1}{r}{0.058} & 
\multicolumn{1}{r}{0.039} & \multicolumn{1}{r}{0.013} \\ 
&  & median $dist$ & \multicolumn{1}{r}{} & \multicolumn{1}{r}{\textbf{0.138}
} & \multicolumn{1}{r}{0.149} & \multicolumn{1}{r}{0.188} \\ 
$SNR=$ & $2.25$ & median $rat$ & \multicolumn{1}{r}{} & \multicolumn{1}{r}{%
\textbf{0.552}} & \multicolumn{1}{r}{0.369} & \multicolumn{1}{r}{0.120} \\ 
\hline
$0.2^{2}$ & $0.03^{2}$ & median $m$ & \multicolumn{1}{r}{0.001} & 
\multicolumn{1}{r}{-0.006} & \multicolumn{1}{r}{-0.006} & \multicolumn{1}{r}{
0.005} \\ 
&  & median $s$ & \multicolumn{1}{r}{0.113} & \multicolumn{1}{r}{0.159} & 
\multicolumn{1}{r}{0.107} & \multicolumn{1}{r}{0.039} \\ 
&  & median $dist$ & \multicolumn{1}{r}{} & \multicolumn{1}{r}{0.157} & 
\multicolumn{1}{r}{0.133} & \multicolumn{1}{r}{\textbf{0.127}} \\ 
$SNR=$ & $0.0225$ & median $rat$ & \multicolumn{1}{r}{} & \multicolumn{1}{r}{
1.598} & \multicolumn{1}{r}{\textbf{1.070}} & \multicolumn{1}{r}{0.366} \\ 
\hline
$1$ & $0.03^{2}$ & median $m$ & \multicolumn{1}{r}{0.000} & 
\multicolumn{1}{r}{-0.006} & \multicolumn{1}{r}{-0.007} & \multicolumn{1}{r}{
0.003} \\ 
&  & median $s$ & \multicolumn{1}{r}{0.113} & \multicolumn{1}{r}{0.318} & 
\multicolumn{1}{r}{0.325} & \multicolumn{1}{r}{0.115} \\ 
&  & median $dist$ & \multicolumn{1}{r}{} & \multicolumn{1}{r}{0.272} & 
\multicolumn{1}{r}{0.280} & \multicolumn{1}{r}{\textbf{0.141}} \\ 
$SNR=$ & $0.0009$ & median $rat$ & \multicolumn{1}{r}{} & \multicolumn{1}{r}{
3.216} & \multicolumn{1}{r}{3.306} & \multicolumn{1}{r}{\textbf{1.149}} \\ 
\hline
$10^{2}$ & $0.03^{2}$ & median $m$ & \multicolumn{1}{r}{0.001} & 
\multicolumn{1}{r}{-0.006} & \multicolumn{1}{r}{-0.005} & \multicolumn{1}{r}{
-0.007} \\ 
&  & median $s$ & \multicolumn{1}{r}{0.113} & \multicolumn{1}{r}{0.402} & 
\multicolumn{1}{r}{0.406} & \multicolumn{1}{r}{0.143} \\ 
&  & median $dist$ & \multicolumn{1}{r}{} & \multicolumn{1}{r}{0.330} & 
\multicolumn{1}{r}{0.337} & \multicolumn{1}{r}{\textbf{0.153}} \\ 
$SNR=$ & $0.00007$ & median $rat$ & \multicolumn{1}{r}{} & 
\multicolumn{1}{r}{4.053} & \multicolumn{1}{r}{4.119} & \multicolumn{1}{r}{%
\textbf{1.437}} \\ \hline\hline
\end{tabular}
{\normalsize 
\begin{minipage}{6.5in}
{\footnotesize
Notes: 1) The numbers in the column under \textquotedblleft True\textquotedblright \ are computed from the data generating process described in Section \ref{alt_tvp_sec4_1}.

2) \textquotedblleft $m$\textquotedblright \ \textquotedblleft $s$\textquotedblright \ \textquotedblleft $dist$\textquotedblright \ \textquotedblleft $rat$\textquotedblright \ stand for the mean, the standard deviation, the distance from the true values, and the ratio of the standard deviation of the estimates to that of the true values of $\beta$.

3) The bold numbers are the smallest (for median \textquotedblleft $dist$\textquotedblright ) and the closest to one (for median \textquotedblleft $rat$\textquotedblright ), indicating the best method out of the three (OLS, 1FGLS, 2FGLS). }
\end{minipage}}
\end{table}

\clearpage

\begin{table}[tbp]
\caption{Simulation Results $T=250$}\label{alt_table2}
\centering
\normalsize
\begin{tabular}{cclllll}
\hline\hline
$H$ & $Q$ &  & True & OLS & 1FGLS & 2FGLS \\ \hline
$0.002^{2}$ & $0.03^{2}$ & median $m$ & \multicolumn{1}{r}{-0.002} & 
\multicolumn{1}{r}{-0.002} & \multicolumn{1}{r}{-0.002} & \multicolumn{1}{r}{
0.005} \\ 
&  & median $s$ & \multicolumn{1}{r}{0.174} & \multicolumn{1}{r}{0.136} & 
\multicolumn{1}{r}{0.115} & \multicolumn{1}{r}{0.029} \\ 
&  & median $dist$ & \multicolumn{1}{r}{} & \multicolumn{1}{r}{\textbf{0.147}
} & \multicolumn{1}{r}{0.164} & \multicolumn{1}{r}{0.337} \\ 
$SNR=$ & $225$ & median $rat$ & \multicolumn{1}{r}{} & \multicolumn{1}{r}{%
\textbf{0.818}} & \multicolumn{1}{r}{0.705} & \multicolumn{1}{r}{0.183} \\ 
\hline
$0.02^{2}$ & $0.03^{2}$ & median $m$ & \multicolumn{1}{r}{-0.005} & 
\multicolumn{1}{r}{-0.003} & \multicolumn{1}{r}{-0.002} & \multicolumn{1}{r}{
0.003} \\ 
&  & median $s$ & \multicolumn{1}{r}{0.173} & \multicolumn{1}{r}{0.143} & 
\multicolumn{1}{r}{0.116} & \multicolumn{1}{r}{0.031} \\ 
&  & median $dist$ & \multicolumn{1}{r}{} & \multicolumn{1}{r}{\textbf{0.125}
} & \multicolumn{1}{r}{0.142} & \multicolumn{1}{r}{0.313} \\ 
$SNR=$ & $2.25$ & median $rat$ & \multicolumn{1}{r}{} & \multicolumn{1}{r}{%
\textbf{0.852}} & \multicolumn{1}{r}{0.699} & \multicolumn{1}{r}{0.191} \\ 
\hline
$0.2^{2}$ & $0.03^{2}$ & median $m$ & \multicolumn{1}{r}{-0.001} & 
\multicolumn{1}{r}{-0.003} & \multicolumn{1}{r}{-0.002} & \multicolumn{1}{r}{
0.001} \\ 
&  & median $s$ & \multicolumn{1}{r}{0.173} & \multicolumn{1}{r}{0.211} & 
\multicolumn{1}{r}{0.169} & \multicolumn{1}{r}{0.058} \\ 
&  & median $dist$ & \multicolumn{1}{r}{} & \multicolumn{1}{r}{0.150} & 
\multicolumn{1}{r}{\textbf{0.130}} & \multicolumn{1}{r}{0.238} \\ 
$SNR=$ & $0.0225$ & median $rat$ & \multicolumn{1}{r}{} & \multicolumn{1}{r}{
1.301} & \multicolumn{1}{r}{\textbf{1.016}} & \multicolumn{1}{r}{0.351} \\ 
\hline
$1$ & $0.03^{2}$ & median $m$ & \multicolumn{1}{r}{-0.000} & 
\multicolumn{1}{r}{-0.003} & \multicolumn{1}{r}{-0.003} & \multicolumn{1}{r}{
-0.003} \\ 
&  & median $s$ & \multicolumn{1}{r}{0.173} & \multicolumn{1}{r}{0.327} & 
\multicolumn{1}{r}{0.326} & \multicolumn{1}{r}{0.112} \\ 
&  & median $dist$ & \multicolumn{1}{r}{} & \multicolumn{1}{r}{0.237} & 
\multicolumn{1}{r}{0.238} & \multicolumn{1}{r}{\textbf{0.192}} \\ 
$SNR=$ & $0.0009$ & median $rat$ & \multicolumn{1}{r}{} & \multicolumn{1}{r}{
2.093} & \multicolumn{1}{r}{2.091} & \multicolumn{1}{r}{\textbf{0.695}} \\ 
\hline
$10^{2}$ & $0.03^{2}$ & median $m$ & \multicolumn{1}{r}{-0.000} & 
\multicolumn{1}{r}{-0.006} & \multicolumn{1}{r}{-0.003} & \multicolumn{1}{r}{
-0.000} \\ 
&  & median $s$ & \multicolumn{1}{r}{0.172} & \multicolumn{1}{r}{0.394} & 
\multicolumn{1}{r}{0.395} & \multicolumn{1}{r}{0.138} \\ 
&  & median $dist$ & \multicolumn{1}{r}{} & \multicolumn{1}{r}{0.282} & 
\multicolumn{1}{r}{0.284} & \multicolumn{1}{r}{\textbf{0.163}} \\ 
$SNR=$ & $0.00007$ & median $rat$ & \multicolumn{1}{r}{} & 
\multicolumn{1}{r}{2.559} & \multicolumn{1}{r}{2.564} & \multicolumn{1}{r}{%
\textbf{0.870}} \\ \hline\hline
\end{tabular}

{\normalsize 
\begin{minipage}{6.5in}
{\footnotesize
Notes: 1) The numbers in the column under \textquotedblleft True\textquotedblright \ are computed from the data generating process described in Section \ref{alt_tvp_sec4_1}.

2) \textquotedblleft $m$\textquotedblright \ \textquotedblleft $s$\textquotedblright \ \textquotedblleft $dist$\textquotedblright \ \textquotedblleft $rat$\textquotedblright \ stand for the mean, the standard deviation, the distance from the true values, and the ratio of the standard deviation of the estimates to that of the true values of $\beta$.

3) The bold numbers are the smallest (for median \textquotedblleft $dist$\textquotedblright ) and the closest to one (for median \textquotedblleft $rat$\textquotedblright ), indicating the best method out of the three (OLS, 1FGLS, 2FGLS). }
\end{minipage}}
\end{table}

\clearpage

\begin{table}[tbp]
\caption{Stochastic Volatility, Autoregressive Stochastic Volatility, and Mixtures of Normals}\label{alt_table3}
\centering
\begin{tabular}{cllllll}
\hline\hline
$T$ & RW/AR &  & True & OLS & 1FGLS & 2FGLS \\ \hline
$100$ & RW & median $m$ & \multicolumn{1}{r}{0.003} & \multicolumn{1}{r}{
-0.004} & \multicolumn{1}{r}{-0.003} & \multicolumn{1}{r}{-0.003} \\ 
&  & median $s$ & \multicolumn{1}{r}{0.136} & \multicolumn{1}{r}{0.310} & 
\multicolumn{1}{r}{0.317} & \multicolumn{1}{r}{0.102} \\ 
&  & median $dist$ & \multicolumn{1}{r}{} & \multicolumn{1}{r}{0.262} & 
\multicolumn{1}{r}{0.270} & \multicolumn{1}{r}{\textbf{0.171}} \\ 
&  & median $rat$ & \multicolumn{1}{r}{} & \multicolumn{1}{r}{2.623} & 
\multicolumn{1}{r}{2.679} & \multicolumn{1}{r}{\textbf{0.850}} \\ \hline
& AR & median $m$ & \multicolumn{1}{r}{0.003} & \multicolumn{1}{r}{-0.004} & 
\multicolumn{1}{r}{-0.003} & \multicolumn{1}{r}{-0.003} \\ 
&  & median $s$ & \multicolumn{1}{r}{0.136} & \multicolumn{1}{r}{0.310} & 
\multicolumn{1}{r}{0.317} & \multicolumn{1}{r}{0.102} \\ 
&  & median $dist$ & \multicolumn{1}{r}{} & \multicolumn{1}{r}{0.262} & 
\multicolumn{1}{r}{0.270} & \multicolumn{1}{r}{\textbf{0.171}} \\ 
&  & median $rat$ & \multicolumn{1}{r}{} & \multicolumn{1}{r}{2.623} & 
\multicolumn{1}{r}{2.680} & \multicolumn{1}{r}{\textbf{0.850}} \\ \hline
$250$ & RW & median $m$ & \multicolumn{1}{r}{-0.002} & \multicolumn{1}{r}{
-0.003} & \multicolumn{1}{r}{-0.003} & \multicolumn{1}{r}{0.000} \\ 
&  & median $s$ & \multicolumn{1}{r}{0.204} & \multicolumn{1}{r}{0.327} & 
\multicolumn{1}{r}{0.327} & \multicolumn{1}{r}{0.103} \\ 
&  & median $dist$ & \multicolumn{1}{r}{} & \multicolumn{1}{r}{\textbf{0.221}
} & \multicolumn{1}{r}{0.223} & \multicolumn{1}{r}{0.264} \\ 
&  & median $rat$ & \multicolumn{1}{r}{} & \multicolumn{1}{r}{\textbf{1.733}}
& \multicolumn{1}{r}{1.734} & \multicolumn{1}{r}{0.542} \\ \hline
& AR & median $m$ & \multicolumn{1}{r}{0.000} & \multicolumn{1}{r}{-0.002} & 
\multicolumn{1}{r}{-0.002} & \multicolumn{1}{r}{0.004} \\ 
&  & median $s$ & \multicolumn{1}{r}{0.205} & \multicolumn{1}{r}{0.326} & 
\multicolumn{1}{r}{0.327} & \multicolumn{1}{r}{0.103} \\ 
&  & median $dist$ & \multicolumn{1}{r}{} & \multicolumn{1}{r}{\textbf{0.221}
} & \multicolumn{1}{r}{0.223} & \multicolumn{1}{r}{0.264} \\ 
&  & median $rat$ & \multicolumn{1}{r}{} & \multicolumn{1}{r}{\textbf{1.730}}
& \multicolumn{1}{r}{1.732} & \multicolumn{1}{r}{0.543} \\ \hline\hline
\end{tabular}

{\normalsize 
\begin{minipage}{6.5in}
{\footnotesize
Notes: 1) The numbers in the column under \textquotedblleft True\textquotedblright \ are computed from the data generating process:
\begin{equation*}
\eta_{it}=\lambda_{t}\zeta_{t}^{1}+\left( 1-\lambda_{t}\right) \zeta
_{t}^{2}
\end{equation*}%
where 
\begin{eqnarray*}
\lambda_{t} &\sim &i.i.d.Bernoulli\left( 0.95\right) \\
\zeta_{1,t} &\sim &\mathcal{N}\left( 0,0.03^{2}\right) ,\text{ \ }\zeta_{2,t}\sim
\mathcal{N}\left( 0,0.1^{2}\right) 
\end{eqnarray*}
and 
\begin{eqnarray*}
\varepsilon_{it} &=&\sqrt{h_{i,t}}\xi_{t} \\
\log h_{i,t} &=&\rho \log h_{i,t-1}+e_{t}
\end{eqnarray*}%
where $\rho=1$ when stochastic volatility is considered (labeled as RW), $\rho=0.9$ when autoregressive stochastic volatility is considered (labeled as AR), $\varepsilon_{it}$ is the $i$-th element of $\varepsilon_{t}$; $\log h_{i,0}=0$, $e_{t}\sim \mathcal{N}\left( 0,0.02^{2}\right) $, and $\xi_{t}\sim \mathcal{N}\left( 0,1\right) $.

2) \textquotedblleft $m$\textquotedblright \ \textquotedblleft $s$\textquotedblright \ \textquotedblleft $dist$\textquotedblright \ \textquotedblleft $rat$\textquotedblright \ stand for the mean, the standard deviation, the distance from the true values, and the ratio of the standard deviation of the estimates to that of the true values of $\beta$.

3) The bold numbers are the smallest (for median \textquotedblleft $dist$\textquotedblright ) and the closest to one (for median \textquotedblleft $rat$\textquotedblright ), indicating the best method out of the three (OLS, 1FGLS, 2FGLS).}
\end{minipage}}
\end{table}

\clearpage

\begin{table}[tbp]
\caption{Mixtures of Normals $T=100$}\label{alt_table4}
\centering
\begin{tabular}{cclllll}
\hline\hline
$H$ & $Q$ &  & True & OLS & 1FGLS & 2FGLS \\ \hline
$0.002^{2}$ & $Mixture$ & median $m$ & \multicolumn{1}{r}{-0.002} & 
\multicolumn{1}{r}{-0.003} & \multicolumn{1}{r}{-0.002} & \multicolumn{1}{r}{
0.001} \\ 
&  & median $s$ & \multicolumn{1}{r}{0.136} & \multicolumn{1}{r}{0.069} & 
\multicolumn{1}{r}{0.056} & \multicolumn{1}{r}{0.015} \\ 
&  & median $dist$ & \multicolumn{1}{r}{} & \multicolumn{1}{r}{\textbf{0.183}
} & \multicolumn{1}{r}{0.194} & \multicolumn{1}{r}{0.262} \\ 
&  & median $rat$ & \multicolumn{1}{r}{} & \multicolumn{1}{r}{\textbf{0.550}}
& \multicolumn{1}{r}{0.436} & \multicolumn{1}{r}{0.123} \\ \hline
$0.02^{2}$ &  & median $m$ & \multicolumn{1}{r}{-0.002} & \multicolumn{1}{r}{
0.005} & \multicolumn{1}{r}{0.005} & \multicolumn{1}{r}{0.003} \\ 
&  & median $s$ & \multicolumn{1}{r}{0.136} & \multicolumn{1}{r}{0.077} & 
\multicolumn{1}{r}{0.058} & \multicolumn{1}{r}{0.017} \\ 
&  & median $dist$ & \multicolumn{1}{r}{} & \multicolumn{1}{r}{\textbf{0.160}
} & \multicolumn{1}{r}{0.171} & \multicolumn{1}{r}{0.242} \\ 
&  & median $rat$ & \multicolumn{1}{r}{} & \multicolumn{1}{r}{\textbf{0.612}}
& \multicolumn{1}{r}{0.468} & \multicolumn{1}{r}{0.137} \\ \hline
$0.2^{2}$ &  & median $m$ & \multicolumn{1}{r}{-0.001} & \multicolumn{1}{r}{
-0.007} & \multicolumn{1}{r}{-0.006} & \multicolumn{1}{r}{0.004} \\ 
&  & median $s$ & \multicolumn{1}{r}{0.136} & \multicolumn{1}{r}{0.166} & 
\multicolumn{1}{r}{0.124} & \multicolumn{1}{r}{0.042} \\ 
&  & median $dist$ & \multicolumn{1}{r}{} & \multicolumn{1}{r}{0.159} & 
\multicolumn{1}{r}{\textbf{0.140}} & \multicolumn{1}{r}{0.165} \\ 
&  & median $rat$ & \multicolumn{1}{r}{} & \multicolumn{1}{r}{1.356} & 
\multicolumn{1}{r}{\textbf{1.006}} & \multicolumn{1}{r}{0.329} \\ \hline
$1$ &  & median $m$ & \multicolumn{1}{r}{-0.001} & \multicolumn{1}{r}{-0.003}
& \multicolumn{1}{r}{-0.005} & \multicolumn{1}{r}{-0.004} \\ 
&  & median $s$ & \multicolumn{1}{r}{0.136} & \multicolumn{1}{r}{0.308} & 
\multicolumn{1}{r}{0.315} & \multicolumn{1}{r}{0.103} \\ 
&  & median $dist$ & \multicolumn{1}{r}{} & \multicolumn{1}{r}{0.260} & 
\multicolumn{1}{r}{0.267} & \multicolumn{1}{r}{\textbf{0.168}} \\ 
&  & median $rat$ & \multicolumn{1}{r}{} & \multicolumn{1}{r}{2.585} & 
\multicolumn{1}{r}{2.641} & \multicolumn{1}{r}{\textbf{0.852}} \\ \hline
$10^{2}$ &  & median $m$ & \multicolumn{1}{r}{-0.001} & \multicolumn{1}{r}{
-0.003} & \multicolumn{1}{r}{-0.003} & \multicolumn{1}{r}{-0.001} \\ 
&  & median $s$ & \multicolumn{1}{r}{0.136} & \multicolumn{1}{r}{0.385} & 
\multicolumn{1}{r}{0.386} & \multicolumn{1}{r}{0.129} \\ 
&  & median $dist$ & \multicolumn{1}{r}{} & \multicolumn{1}{r}{0.311} & 
\multicolumn{1}{r}{0.316} & \multicolumn{1}{r}{\textbf{0.163}} \\ 
&  & median $rat$ & \multicolumn{1}{r}{} & \multicolumn{1}{r}{3.237} & 
\multicolumn{1}{r}{3.261} & \multicolumn{1}{r}{\textbf{1.078}} \\ 
\hline\hline
\end{tabular}

{\normalsize 
\begin{minipage}{6.5in}
{\footnotesize
Notes: 1) The numbers in the column under \textquotedblleft True\textquotedblright \ are computed from the data generating process:
\begin{equation*}
\eta_{it}=\lambda_{t}\zeta_{t}^{1}+\left( 1-\lambda_{t}\right) \zeta_{t}^{2}
\end{equation*}%
where 
\begin{eqnarray*}
\lambda_{t} &\sim &i.i.d.Bernoulli\left( 0.95\right) \\
\zeta_{1,t} &\sim &\mathcal{N}\left( 0,0.03^{2}\right) ,\text{ \ }\zeta_{2,t}\sim
\mathcal{N}\left( 0,0.1^{2}\right). 
\end{eqnarray*}

 2) \textquotedblleft $m$\textquotedblright \ \textquotedblleft $s$\textquotedblright \ \textquotedblleft $dist$\textquotedblright \ \textquotedblleft $rat$\textquotedblright \ stand for the mean, the standard deviation, the distance from the true values, and the ratio of the standard deviation of the estimates to that of the true values of $\beta$.

3) The bold numbers are the smallest (for median \textquotedblleft $dist$\textquotedblright ) and the closest to one (for median \textquotedblleft $rat$\textquotedblright ), indicating the best method out of the three (OLS, 1FGLS, 2FGLS).}
\end{minipage}}
\end{table}

\clearpage

\begin{table}[tbp]
\caption{Mixtures of Normals $T=250$}\label{alt_table5}
\centering
\begin{tabular}{cclllll}
\hline\hline
$H$ & $Q$ &  & True & OLS & 1FGLS & 2FGLS \\ \hline
$0.002^{2}$ & $Mixture$ & median $m$ & \multicolumn{1}{r}{-0.004} & 
\multicolumn{1}{r}{0.004} & \multicolumn{1}{r}{0.004} & \multicolumn{1}{r}{
0.000} \\ 
&  & median $s$ & \multicolumn{1}{r}{0.205} & \multicolumn{1}{r}{0.166} & 
\multicolumn{1}{r}{0.145} & \multicolumn{1}{r}{0.039} \\ 
&  & median $dist$ & \multicolumn{1}{r}{} & \multicolumn{1}{r}{\textbf{0.150}
} & \multicolumn{1}{r}{0.164} & \multicolumn{1}{r}{0.375} \\ 
&  & median $rat$ & \multicolumn{1}{r}{} & \multicolumn{1}{r}{\textbf{0.843}}
& \multicolumn{1}{r}{0.747} & \multicolumn{1}{r}{0.204} \\ \hline
$0.02^{2}$ &  & median $m$ & \multicolumn{1}{r}{-0.003} & \multicolumn{1}{r}{
0.003} & \multicolumn{1}{r}{0.002} & \multicolumn{1}{r}{-0.002} \\ 
&  & median $s$ & \multicolumn{1}{r}{0.206} & \multicolumn{1}{r}{0.173} & 
\multicolumn{1}{r}{0.149} & \multicolumn{1}{r}{0.042} \\ 
&  & median $dist$ & \multicolumn{1}{r}{} & \multicolumn{1}{r}{\textbf{0.136}
} & \multicolumn{1}{r}{0.150} & \multicolumn{1}{r}{0.365} \\ 
&  & median $rat$ & \multicolumn{1}{r}{} & \multicolumn{1}{r}{\textbf{0.866}}
& \multicolumn{1}{r}{0.755} & \multicolumn{1}{r}{0.214} \\ \hline
$0.2^{2}$ &  & median $m$ & \multicolumn{1}{r}{0.000} & \multicolumn{1}{r}{
-0.002} & \multicolumn{1}{r}{-0.001} & \multicolumn{1}{r}{0.004} \\ 
&  & median $s$ & \multicolumn{1}{r}{0.203} & \multicolumn{1}{r}{0.227} & 
\multicolumn{1}{r}{0.193} & \multicolumn{1}{r}{0.062} \\ 
&  & median $dist$ & \multicolumn{1}{r}{} & \multicolumn{1}{r}{0.148} & 
\multicolumn{1}{r}{\textbf{0.135}} & \multicolumn{1}{r}{0.303} \\ 
&  & median $rat$ & \multicolumn{1}{r}{} & \multicolumn{1}{r}{1.172} & 
\multicolumn{1}{r}{\textbf{0.976}} & \multicolumn{1}{r}{0.316} \\ \hline
$1$ &  & median $m$ & \multicolumn{1}{r}{0.001} & \multicolumn{1}{r}{0.002}
& \multicolumn{1}{r}{0.001} & \multicolumn{1}{r}{0.001} \\ 
&  & median $s$ & \multicolumn{1}{r}{0.202} & \multicolumn{1}{r}{0.326} & 
\multicolumn{1}{r}{0.326} & \multicolumn{1}{r}{0.105} \\ 
&  & median $dist$ & \multicolumn{1}{r}{} & \multicolumn{1}{r}{\textbf{0.223}
} & \multicolumn{1}{r}{0.224} & \multicolumn{1}{r}{0.245} \\ 
&  & median $rat$ & \multicolumn{1}{r}{} & \multicolumn{1}{r}{\textbf{1.755}}
& \multicolumn{1}{r}{1.755} & \multicolumn{1}{r}{0.552} \\ \hline
$10^{2}$ &  & median $m$ & \multicolumn{1}{r}{0.001} & \multicolumn{1}{r}{
0.001} & \multicolumn{1}{r}{0.000} & \multicolumn{1}{r}{0.001} \\ 
&  & median $s$ & \multicolumn{1}{r}{0.201} & \multicolumn{1}{r}{0.385} & 
\multicolumn{1}{r}{0.385} & \multicolumn{1}{r}{0.129} \\ 
&  & median $dist$ & \multicolumn{1}{r}{} & \multicolumn{1}{r}{0.260} & 
\multicolumn{1}{r}{0.262} & \multicolumn{1}{r}{\textbf{0.204}} \\ 
&  & median $rat$ & \multicolumn{1}{r}{} & \multicolumn{1}{r}{2.106} & 
\multicolumn{1}{r}{2.104} & \multicolumn{1}{r}{\textbf{0.689}} \\ 
\hline\hline
\end{tabular}

{\normalsize 
\begin{minipage}{6.5in}
{\footnotesize
Notes: 1) The numbers in the column under \textquotedblleft True\textquotedblright \ are computed from the data generating process:
\begin{equation*}
\eta_{it}=\lambda_{t}\zeta_{t}^{1}+\left( 1-\lambda_{t}\right) \zeta
_{t}^{2}
\end{equation*}%
where 
\begin{eqnarray*}
\lambda_{t} &\sim &i.i.d.Bernoulli\left( 0.95\right) \\
\zeta_{1,t} &\sim &\mathcal{N}\left( 0,0.03^{2}\right) ,\text{ \ }\zeta_{2,t}\sim
\mathcal{N}\left( 0,0.1^{2}\right). 
\end{eqnarray*}

 2) \textquotedblleft $m$\textquotedblright \ \textquotedblleft $s$\textquotedblright \ \textquotedblleft $dist$\textquotedblright \ \textquotedblleft $rat$\textquotedblright \ stand for the mean, the standard deviation, the distance from the true values, and the ratio of the standard deviation of the estimates to that of the true values of $\beta$.

3) The bold numbers are the smallest (for median \textquotedblleft $dist$\textquotedblright ) and the closest to one (for median \textquotedblleft $rat$\textquotedblright ), indicating the best method out of the three (OLS, 1FGLS, 2FGLS). }
\end{minipage}}
\end{table}

\clearpage

\begin{table}[tbp]
\caption{Stochastic Volatility and Autoregressive Stochastic Volatility $T=100$}\label{alt_table6}
\centering
\begin{tabular}{cclllll}
\hline\hline
$T$ & $Q$ &  & True & OLS & 1FGLS & 2FGLS \\ \hline
$100$ & $0.03^{2}$ & median $m$ & \multicolumn{1}{r}{-0.002} & 
\multicolumn{1}{r}{-0.010} & \multicolumn{1}{r}{-0.011} & \multicolumn{1}{r}{
-0.011} \\ 
& RW & median $s$ & \multicolumn{1}{r}{0.113} & \multicolumn{1}{r}{0.317} & 
\multicolumn{1}{r}{0.324} & \multicolumn{1}{r}{0.113} \\ 
&  & median $dist$ & \multicolumn{1}{r}{} & \multicolumn{1}{r}{0.272} & 
\multicolumn{1}{r}{0.279} & \multicolumn{1}{r}{\textbf{0.141}} \\ 
&  & median $rat$ & \multicolumn{1}{r}{} & \multicolumn{1}{r}{3.218} & 
\multicolumn{1}{r}{3.298} & \multicolumn{1}{r}{\textbf{1.150}} \\ \hline
$100$ & $0.03^{2}$ & median $m$ & \multicolumn{1}{r}{-0.002} & 
\multicolumn{1}{r}{-0.010} & \multicolumn{1}{r}{-0.010} & \multicolumn{1}{r}{
-0.011} \\ 
& AR & median $s$ & \multicolumn{1}{r}{0.113} & \multicolumn{1}{r}{0.317} & 
\multicolumn{1}{r}{0.324} & \multicolumn{1}{r}{0.113} \\ 
&  & median $dist$ & \multicolumn{1}{r}{} & \multicolumn{1}{r}{0.272} & 
\multicolumn{1}{r}{0.279} & \multicolumn{1}{r}{\textbf{0.141}} \\ 
&  & median $rat$ & \multicolumn{1}{r}{} & \multicolumn{1}{r}{3.217} & 
\multicolumn{1}{r}{3.298} & \multicolumn{1}{r}{\textbf{1.150}} \\ \hline
$250$ & $0.03^{2}$ & median $m$ & \multicolumn{1}{r}{-0.003} & 
\multicolumn{1}{r}{-0.008} & \multicolumn{1}{r}{-0.009} & \multicolumn{1}{r}{
-0.005} \\ 
& RW & median $s$ & \multicolumn{1}{r}{0.174} & \multicolumn{1}{r}{0.327} & 
\multicolumn{1}{r}{0.328} & \multicolumn{1}{r}{0.113} \\ 
&  & median $dist$ & \multicolumn{1}{r}{} & \multicolumn{1}{r}{0.236} & 
\multicolumn{1}{r}{0.238} & \multicolumn{1}{r}{\textbf{0.182}} \\ 
&  & median $rat$ & \multicolumn{1}{r}{} & \multicolumn{1}{r}{2.084} & 
\multicolumn{1}{r}{2.093} & \multicolumn{1}{r}{\textbf{0.703}} \\ \hline
$250$ & $0.03^{2}$ & median $m$ & \multicolumn{1}{r}{-0.003} & 
\multicolumn{1}{r}{-0.010} & \multicolumn{1}{r}{-0.010} & \multicolumn{1}{r}{
0.004} \\ 
& AR & median $s$ & \multicolumn{1}{r}{0.174} & \multicolumn{1}{r}{0.327} & 
\multicolumn{1}{r}{0.327} & \multicolumn{1}{r}{0.113} \\ 
&  & median $dist$ & \multicolumn{1}{r}{} & \multicolumn{1}{r}{0.236} & 
\multicolumn{1}{r}{0.238} & \multicolumn{1}{r}{\textbf{0.182}} \\ 
&  & median $rat$ & \multicolumn{1}{r}{} & \multicolumn{1}{r}{2.085} & 
\multicolumn{1}{r}{2.092} & \multicolumn{1}{r}{\textbf{0.702}} \\ 
\hline\hline
\end{tabular}

{\normalsize 
\begin{minipage}{6.5in}
{\footnotesize
Notes: 1) The numbers in the column under \textquotedblleft True\textquotedblright \ are computed from the data generating process:
\begin{eqnarray*}
\varepsilon_{it} &=&\sqrt{h_{i,t}}\xi_{t} \\
\log h_{i,t} &=&\rho \log h_{i,t-1}+e_{t}
\end{eqnarray*}%
where $\rho =1$ when stochastic volatility is considered (labeled as RW), $\rho=0.9$ when autoregressive stochastic volatility is considered (labeled as AR), $\varepsilon_{it}$ is the $i$-th element of $\varepsilon_{t}$; $\log h_{i,0}=0$, $e_{t}\sim \mathcal{N}\left( 0,0.02^{2}\right) $, and $\xi_{t}\sim \mathcal{N}\left( 0,1\right)$.

 2) \textquotedblleft $m$\textquotedblright \ \textquotedblleft $s$\textquotedblright \ \textquotedblleft $dist$\textquotedblright \ \textquotedblleft $rat$\textquotedblright \ stand for the mean, the standard deviation, the distance from the true values, and the ratio of the standard deviation of the estimates to that of the true values of $\beta$.

3) The bold numbers are the smallest (for median \textquotedblleft $dist$\textquotedblright) and the closest to one (for median \textquotedblleft $rat$\textquotedblright), indicating the best method out of the three (OLS, 1FGLS, 2FGLS). }
\end{minipage}}
\end{table}

\clearpage

\begin{figure}[bp]
 \caption*{Figures 1 through 4: The Estimated Time-Varying Parameters}
  \centering
    \begin{tabular}{c}
 
      \begin{minipage}{0.47\hsize}
        \centering
          \includegraphics[keepaspectratio, scale=0.50, angle=0]
                          {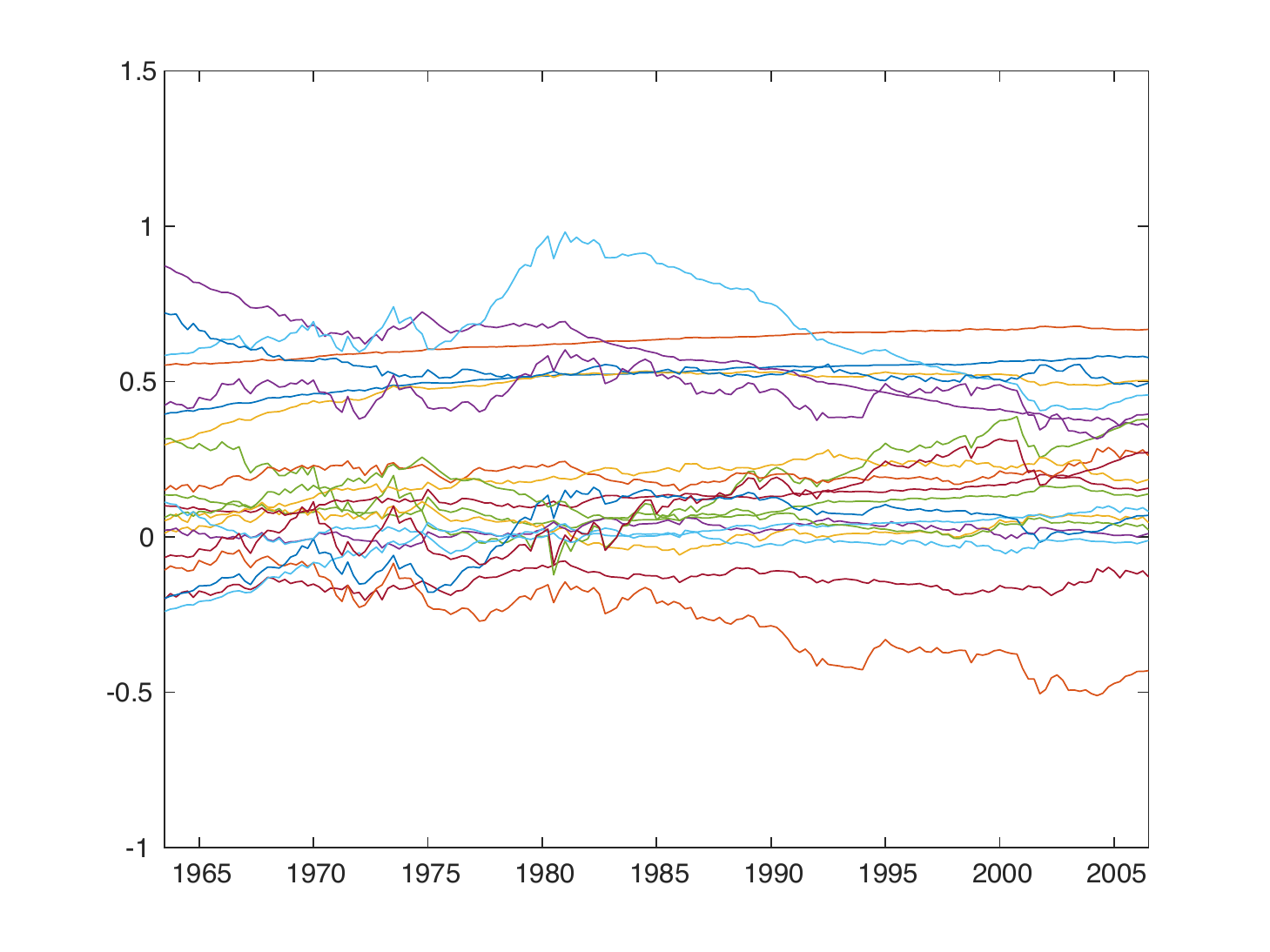}
			  \caption{OLS}
                          \label{alt_fig1}
      \end{minipage}
  
      \begin{minipage}{0.47\hsize}
        \centering
          \includegraphics[keepaspectratio, scale=0.50, angle=0]
                          {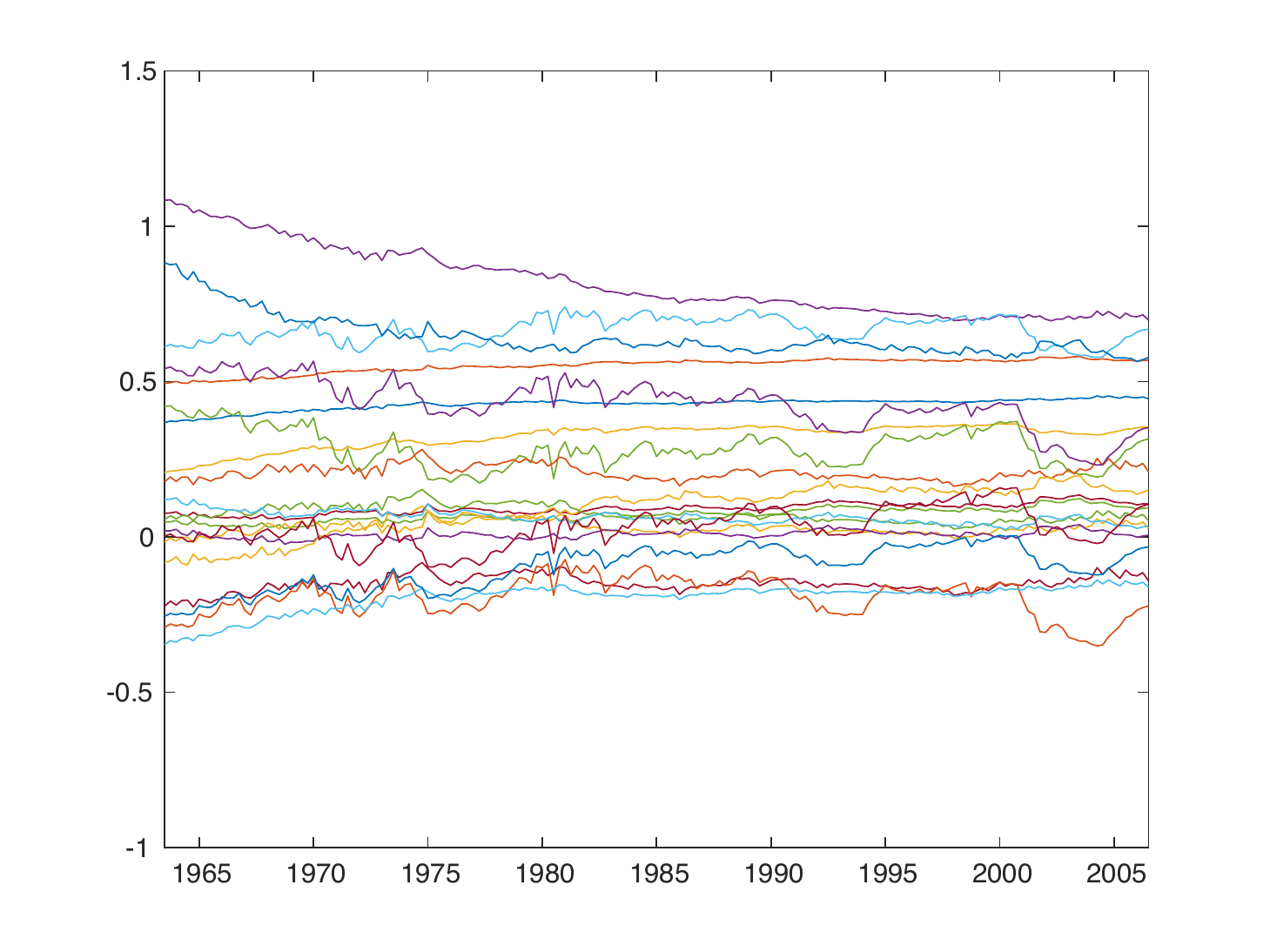}
			  \caption{1FGLS}
                          \label{alt_fig2}
      \end{minipage} \\
 
 
      \begin{minipage}{0.47\hsize}
        \centering
          \includegraphics[keepaspectratio, scale=0.50, angle=0]
                          {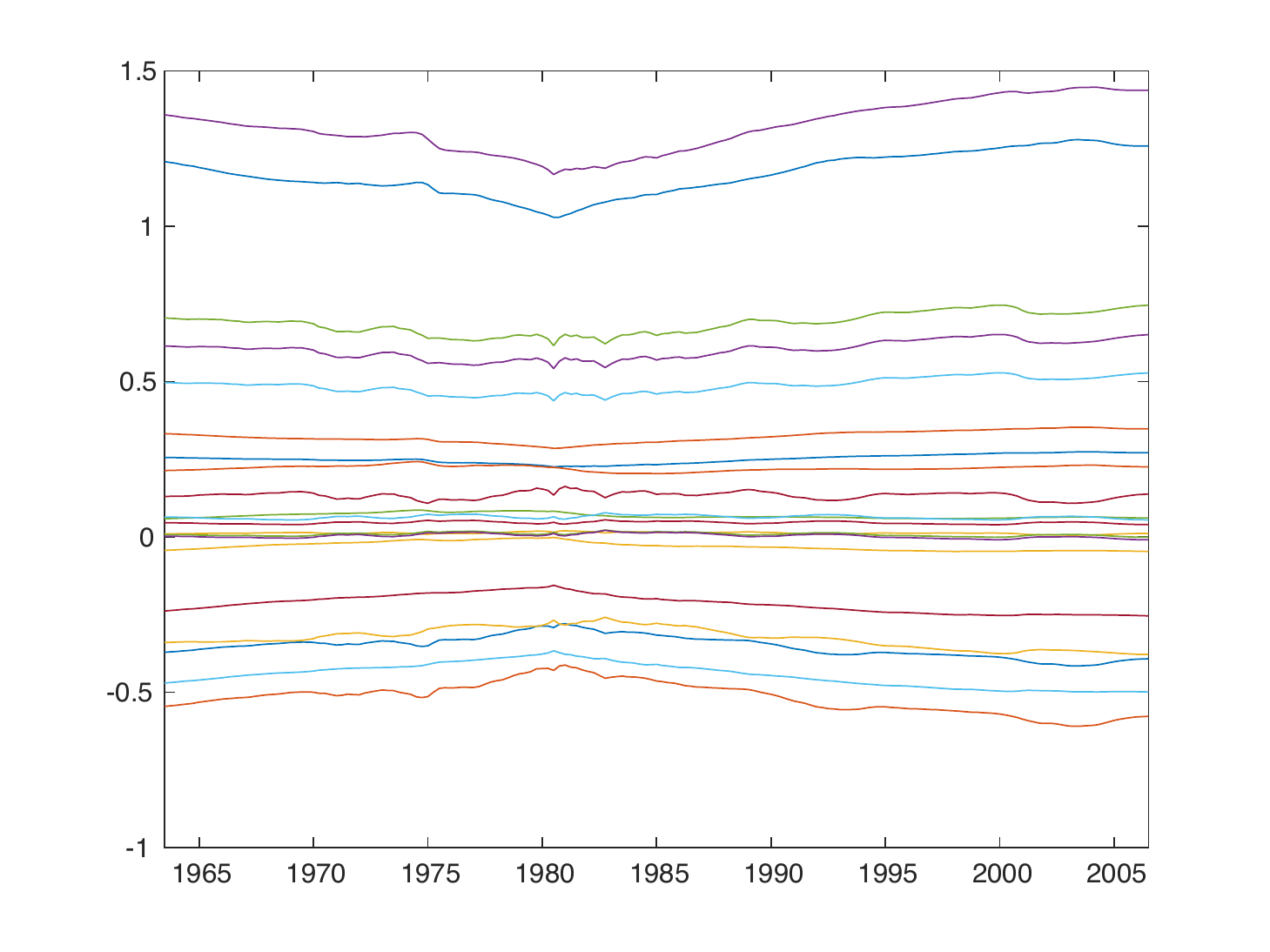}
			  \caption{2FGLS}
                          \label{alt_fig3}
      \end{minipage} 
 
 
      \begin{minipage}{0.47\hsize}
        \centering
          \includegraphics[keepaspectratio, scale=0.50, angle=0]
                          {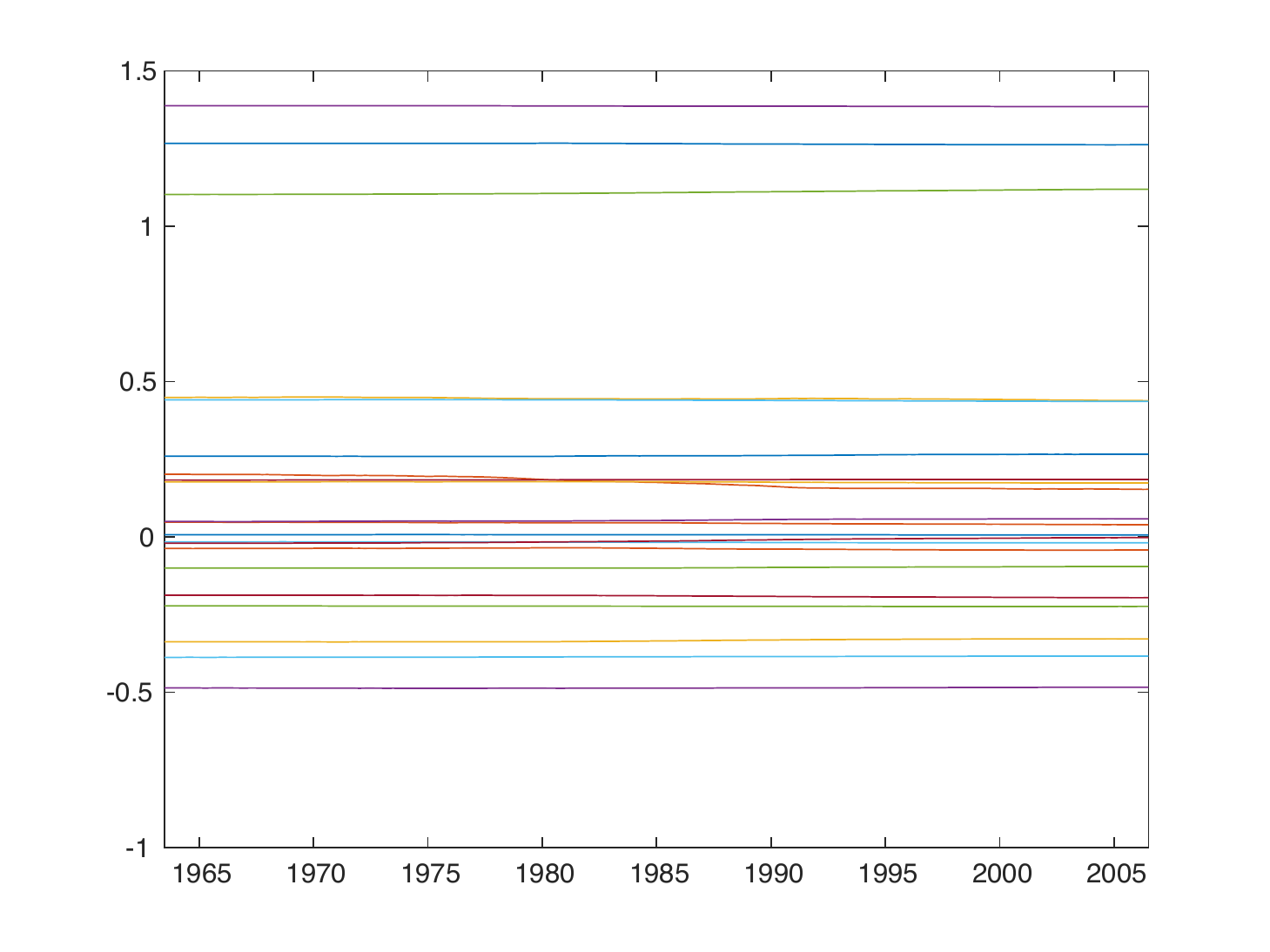}
			  \caption{Bayesian}
                          \label{alt_fig4}
      \end{minipage}

    \end{tabular}
\end{figure}

%% file: alt_tvp_appendix.tex
\section*{Appendix 1: Proof of Propositions}

\begin{itemize}
\item Proof of Proposition 1

\begin{lemma}
$\left( S-TU^{-1}V\right) ^{-1}=S^{-1}+S^{-1}T\left( U-VS^{-1}T\right)
^{-1}VS^{-1}$ provided $S^{-1}$ exists.
\end{lemma}

The GLS estimate of $\beta $ is 
\begin{eqnarray}
\widehat{\beta }_{GLS} &=&\left[ \left[ 
\begin{array}{cc}
Z^{\prime }H^{^{\prime }-1/2} & C^{\prime -1}Q^{^{\prime }-1/2}%
\end{array}%
\right] \left[ 
\begin{array}{c}
H^{-1/2}Z \\ 
Q^{-1/2}C^{-1}%
\end{array}%
\right] \right] ^{-1}  \notag \\
&&\times \left[ 
\begin{array}{cc}
Z^{\prime }H^{^{\prime }-1/2} & -C^{\prime -1}Q^{^{\prime }-1/2}%
\end{array}%
\right] \left[ 
\begin{array}{c}
H^{-1/2}Y_{T} \\ 
-Q^{-1/2}b_{0}^{\ast }%
\end{array}%
\right] \\
&=&\left( Z^{\prime }H^{-1}Z+C^{\prime -1}Q^{-1}C^{-1}\right) ^{-1}\left(
Z^{\prime }H^{-1}Y_{T}+C^{-1\prime }Q^{-1}b_{0}^{\ast }\right)  \notag \\
&=&\left[ CQC^{\prime }-CQC^{\prime }Z^{\prime }\Omega ^{-1}ZCQC^{\prime }%
\right] \left( Z^{\prime }H^{-1}Y_{T}+C^{-1\prime }Q^{-1}b_{0}^{\ast }\right)
\notag \\
&=&CQC^{\prime }Z^{\prime }\Omega ^{-1}Y_{T}+\left[ C-CQC^{\prime }Z^{\prime
}\Omega ^{-1}ZC\right] b_{0}^{\ast }  \notag \\
&=&Cb_{0}^{\ast }+CQC^{\prime }Z^{\prime }\Omega ^{-1}\left(
Y_{T}-ZCb_{0}^{\ast }\right) .  \label{beta_wi}
\end{eqnarray}%
Here, we used Lemma, 
\[
 \left( Z^{\prime }H^{-1}Z+C^{\prime -1}Q^{-1}C^{-1}\right)^{-1}=CQC^{\prime }-CQC^{\prime }Z^{\prime }\left( H+Z^{\ast }CQC^{\prime}Z^{\prime }\right) ^{-1}ZCQC^{\prime }.
\]
From (\ref{beta_wi}), the conditional variance of $\widehat{\beta }_{GLS}$ is
\begin{eqnarray*}
Var\left( \widehat{\beta }_{GLS}|Y_{T}\right) &=&\left( Z^{\prime
}H^{-1}Z+C^{\prime -1}Q^{-1}C^{-1}\right) ^{-1} \\
&=&CQC^{\prime }-CQC^{\prime }Z^{\prime }\Omega ^{-1}ZCQC^{\prime }.
\end{eqnarray*}
\end{itemize}

\begin{itemize}
\item Detailed Proof of Proposition 2

\begin{lemma}
If $G^{-1}$ and the inverse of $F=A-BG^{-1}E$ exist,%
\begin{equation*}
\left( 
\begin{array}{cc}
A & B \\ 
E & G%
\end{array}%
\right) ^{-1}=\left( 
\begin{array}{cc}
F^{-1} & -F^{-1}BG^{-1} \\ 
-G^{-1}EF^{-1} & G^{-1}+G^{-1}EF^{-1}BG^{-1}%
\end{array}%
\right) .
\end{equation*}
\end{lemma}
\end{itemize}

In our case, 
\begin{equation*}
A=\mathcal{I}^{\prime }H^{-1}\mathcal{I}\text{, \ }B=\mathcal{I}^{\prime
}H^{-1}Z\text{, }G=Z^{\prime }H^{-1}Z+C^{\prime -1}Q^{-1}C^{-1}\text{, }%
E=Z^{\prime }H^{-1}\mathcal{I}\text{;}
\end{equation*}%
and 
\begin{eqnarray*}
F &=&\mathcal{I}^{\prime }H^{-1}\mathcal{I-I}^{\prime }H^{-1}Z\left(
Z^{\prime }H^{-1}Z+C^{\prime -1}Q^{-1}C^{-1}\right) ^{-1}Z^{\prime }H^{-1}%
\mathcal{I} \\
&=&\mathcal{I}^{\prime }\left[ H^{-1}\mathcal{-}H^{-1}Z\left( Z^{\prime
}H^{-1}Z+C^{\prime -1}Q^{-1}C^{-1}\right) ^{-1}Z^{\prime }H^{-1}\right] 
\mathcal{I},
\end{eqnarray*}

whose inverse is
\begin{eqnarray*}
F^{-1} &=&\left\{ \mathcal{I}^{\prime }\left[ H^{-1}\mathcal{-}H^{-1}Z\left(
Z^{\prime }H^{-1}Z+C^{\prime -1}Q^{-1}C^{-1}\right) ^{-1}Z^{\prime }H^{-1}%
\right] \mathcal{I}\right\} ^{-1} \\
&=&\left[ \mathcal{I}^{\prime }\left( H+ZCQC^{\prime }Z^{\prime }\right)
^{-1}\mathcal{I}\right] ^{-1} \\
&=&\left( \mathcal{I}^{\prime }\Omega ^{-1}\mathcal{I}\right) ^{-1}.
\end{eqnarray*}

Other useful equations are 
\begin{eqnarray*}
G^{-1} &=&\left( Z^{\prime }H^{-1}Z+C^{\prime -1}Q^{-1}C^{-1}\right) ^{-1} \\
&=&CQC^{\prime }-CQC^{\prime }Z^{\prime }\left( H+ZCQC^{\prime }Z^{\prime
}\right) ^{-1}ZCQC^{\prime } \\
&=&CQC^{\prime }-CQC^{\prime }Z^{\prime }\Omega ^{-1}ZCQC^{\prime };
\end{eqnarray*}
\begin{eqnarray}
\Omega ^{-1} &=&\left( H+ZCQC^{\prime }Z^{\prime }\right) ^{-1}  \notag \\
&=&H^{-1}\mathcal{-}H^{-1}Z\left( Z^{\prime }H^{-1}Z+C^{\prime
-1}Q^{-1}C^{-1}\right) ^{-1}Z^{\prime }H^{-1}  \notag \\
&=&H^{-1}\mathcal{-}H^{-1}ZG^{-1}Z^{\prime }H^{-1}.  \label{Omi}
\end{eqnarray}

Then, for (\ref{GLS}), we arrive at 
\begin{eqnarray}
\widehat{v} &=&\underset{1}{\underbrace{\left( F^{-1}\mathcal{I}^{\prime
}H^{-1}-F^{-1}BG^{-1}Z^{\prime }H^{-1}\right) }}Y_{T}-\underset{2}{%
\underbrace{F^{-1}BG^{-1}C^{-1\prime }Q^{-1}}}b_{0}^{\ast }  \label{vhat} \\
\widehat{\beta } &=&\underset{3}{\underbrace{\left[ -G^{-1}EF^{-1}\mathcal{I}%
^{\prime }H^{-1}+\left( G^{-1}+G^{-1}EF^{-1}BG^{-1}\right) Z^{\prime }H^{-1}%
\right] }}Y_{T}  \notag \\
&&+\underset{4}{\underbrace{\left( G^{-1}+G^{-1}EF^{-1}BG^{-1}\right)
C^{-1\prime }Q^{-1}}}b_{0}^{\ast }.  \label{ahat}
\end{eqnarray}

1.
\begin{eqnarray*}
F^{-1}\mathcal{I}^{\prime }H^{-1}-F^{-1}BG^{-1}Z^{\prime }H^{-1} &=&\left( 
\mathcal{I}^{\prime }\Omega ^{-1}\mathcal{I}\right) ^{-1}\mathcal{I}^{\prime
}H^{-1} \\
&&\times \left( I-ZCQC^{\prime }Z^{\prime }H^{-1}+ZCQC^{\prime }Z^{\prime
}\Omega ^{-1}ZCQC^{\prime }Z^{\prime }H^{-1}\right) \\
&=&\left( \mathcal{I}^{\prime }\Omega ^{-1}\mathcal{I}\right) ^{-1}\mathcal{I%
}^{\prime }H^{-1}\left[ I-\left( \Omega -H\right) H^{-1}+\left( \Omega
-H\right) \Omega ^{-1}\left( \Omega -H\right) H^{-1}\right] \\
&=&\left( \mathcal{I}^{\prime }\Omega ^{-1}\mathcal{I}\right) ^{-1}\mathcal{I%
}^{\prime }H^{-1}\left( I-\Omega H^{-1}+I+\Omega H^{-1}-I-I+H\Omega
^{-1}\right) \\
&=&\left( \mathcal{I}^{\prime }\Omega ^{-1}\mathcal{I}\right) ^{-1}\mathcal{I%
}^{\prime }\Omega ^{-1}.
\end{eqnarray*}

2. 
\begin{eqnarray*}
F^{-1}BG^{-1}C^{-1\prime }Q^{-1} &=&\left( \mathcal{I}^{\prime }\Omega ^{-1}%
\mathcal{I}\right) ^{-1}\mathcal{I}^{\prime }H^{-1}Z\left( CQC^{\prime
}-CQC^{\prime }Z^{\prime }\Omega ^{-1}ZCQC^{\prime }\right) C^{-1\prime
}Q^{-1} \\
&=&\left( \mathcal{I}^{\prime }\Omega ^{-1}\mathcal{I}\right) ^{-1}\mathcal{I%
}^{\prime }H^{-1}Z\left( C-CQC^{\prime }Z^{\prime }\Omega ^{-1}ZC\right) \\
&=&\left( \mathcal{I}^{\prime }\Omega ^{-1}\mathcal{I}\right) ^{-1}\mathcal{I%
}^{\prime }H^{-1}\left( ZC-ZCQC^{\prime }Z^{\prime }\Omega ^{-1}ZC\right) \\
&=&\left( \mathcal{I}^{\prime }\Omega ^{-1}\mathcal{I}\right) ^{-1}\mathcal{I%
}^{\prime }H^{-1}\left( I-ZCQC^{\prime }Z^{\prime }\Omega ^{-1}\right) ZC \\
&=&\left( \mathcal{I}^{\prime }\Omega ^{-1}\mathcal{I}\right) ^{-1}\mathcal{I%
}^{\prime }H^{-1}\left( I-\left( \Omega -H\right) \Omega ^{-1}\right) ZC \\
&=&\left( \mathcal{I}^{\prime }\Omega ^{-1}\mathcal{I}\right) ^{-1}\mathcal{I%
}^{\prime }\Omega ^{-1}ZC.
\end{eqnarray*}

Therefore, 
\begin{eqnarray*}
\widehat{v} &=&\left( \mathcal{I}^{\prime }\Omega ^{-1}\mathcal{I}\right)
^{-1}\mathcal{I}^{\prime }\Omega ^{-1}Y_{T}-\left( \mathcal{I}^{\prime
}\Omega ^{-1}\mathcal{I}\right) ^{-1}\mathcal{I}^{\prime }\Omega
^{-1}ZCb_{0}^{\ast } \\
&=&\left( \mathcal{I}^{\prime }\Omega ^{-1}\mathcal{I}\right) ^{-1}\mathcal{I%
}^{\prime }\Omega ^{-1}\left( Y_{T}-ZCb_{0}^{\ast }\right)
\end{eqnarray*}

3. 
\begin{eqnarray*}
&&-G^{-1}EF^{-1}\mathcal{I}^{\prime }H^{-1}+\left(
G^{-1}+G^{-1}EF^{-1}BG^{-1}\right) Z^{\prime }H^{-1} \\
&=&-G^{-1}EF^{-1}\left( \mathcal{I}^{\prime }H^{-1}-BG^{-1}Z^{\prime
}H^{-1}\right) +G^{-1}Z^{\prime }H^{-1} \\
&=&-G^{-1}EF^{-1}\mathcal{I}^{\prime }\left( H^{-1}-H^{-1}ZG^{-1}Z^{\prime
}H^{-1}\right) +G^{-1}Z^{\prime }H^{-1} \\
&=&-G^{-1}EF^{-1}\mathcal{I}^{\prime }\Omega ^{-1}+G^{-1}Z^{\prime }H^{-1}%
\text{ \ \ \ \ (from \ref{Omi})} \\
&=&-G^{-1}Z^{\prime }H^{-1}\mathcal{I}F^{-1}\mathcal{I}^{\prime }\Omega
^{-1}+G^{-1}Z^{\prime }H^{-1} \\
&=&-G^{-1}Z^{\prime }H^{-1}\left( I-\mathcal{I}F^{-1}\mathcal{I}^{\prime
}\Omega ^{-1}\right) \\
&=&-\left( CQC^{\prime }-CQC^{\prime }Z^{\prime }\Omega ^{-1}ZCQC^{\prime
}\right) Z^{\prime }H^{-1}\left( I-\mathcal{I}F^{-1}\mathcal{I}^{\prime
}\Omega ^{-1}\right) \\
&=&-CQC^{\prime }Z^{\prime }H^{-1}\left( I-\mathcal{I}F^{-1}\mathcal{I}%
^{\prime }\Omega ^{-1}\right) +CQC^{\prime }Z^{\prime }\Omega ^{-1}\left(
\Omega -H\right) H^{-1}\left( I-\mathcal{I}F^{-1}\mathcal{I}^{\prime }\Omega
^{-1}\right) \\
&=&CQC^{\prime }Z^{\prime }\Omega ^{-1}\left( I-\mathcal{I}F^{-1}\mathcal{I}%
^{\prime }\Omega ^{-1}\right) \\
&=&CQC^{\prime }Z^{\prime }\Omega ^{-1}-CQC^{\prime }Z^{\prime }\Omega ^{-1}%
\mathcal{I}\left( \mathcal{I}^{\prime }\Omega ^{-1}\mathcal{I}\right) ^{-1}%
\mathcal{I}^{\prime }\Omega ^{-1} \\
&=&CQC^{\prime }Z^{\prime }\Omega ^{-1}\left[ I-\mathcal{I}\left( \mathcal{I}%
^{\prime }\Omega ^{-1}\mathcal{I}\right) ^{-1}\mathcal{I}^{\prime }\Omega
^{-1}\right]
\end{eqnarray*}

4. 
\begin{eqnarray*}
&&\left( G^{-1}+G^{-1}EF^{-1}BG^{-1}\right) C^{-1\prime }Q^{-1} \\
&=&G^{-1}C^{-1\prime }Q^{-1}+G^{-1}EF^{-1}BG^{-1}C^{-1\prime }Q^{-1} \\
&=&\left( CQC^{\prime }-CQC^{\prime }Z^{\prime }\Omega ^{-1}ZCQC^{\prime
}\right) C^{-1\prime }Q^{-1}+G^{-1}E\left( \mathcal{I}^{\prime }\Omega ^{-1}%
\mathcal{I}\right) ^{-1}\mathcal{I}^{\prime }\Omega ^{-1}ZC\text{ \ (from 2) 
} \\
&=&C-CQC^{\prime }Z^{\prime }\Omega ^{-1}ZC \\
&&+\left( CQC^{\prime }-CQC^{\prime }Z^{\prime }\Omega ^{-1}ZCQC^{\prime
}\right) Z^{\prime }H^{-1}\mathcal{I}\left( \mathcal{I}^{\prime }\Omega ^{-1}%
\mathcal{I}\right) ^{-1}\mathcal{I}^{\prime }\Omega ^{-1}ZC\text{ } \\
&=&C-CQC^{\prime }Z^{\prime }\Omega ^{-1}ZC \\
&&+CQC^{\prime }\left( Z^{\prime }H^{-1}-Z^{\prime }\Omega ^{-1}ZCQC^{\prime
}Z^{\prime }H^{-1}\right) \mathcal{I}\left( \mathcal{I}^{\prime }\Omega ^{-1}%
\mathcal{I}\right) ^{-1}\mathcal{I}^{\prime }\Omega ^{-1}ZC\text{ } \\
&=&C-CQC^{\prime }Z^{\prime }\Omega ^{-1}ZC \\
&&+CQC^{\prime }\left[ Z^{\prime }H^{-1}-Z^{\prime }\Omega ^{-1}\left(
\Omega -H\right) H^{-1}\right] \mathcal{I}\left( \mathcal{I}^{\prime }\Omega
^{-1}\mathcal{I}\right) ^{-1}\mathcal{I}^{\prime }\Omega ^{-1}ZC\text{ } \\
&=&C-CQC^{\prime }Z^{\prime }\Omega ^{-1}ZC+CQC^{\prime }Z^{\prime }\Omega
^{-1}\mathcal{I}\left( \mathcal{I}^{\prime }\Omega ^{-1}\mathcal{I}\right)
^{-1}\mathcal{I}^{\prime }\Omega ^{-1}ZC\text{ } \\
&=&\left[ I-CQC^{\prime }Z^{\prime }\Omega ^{-1}Z+CQC^{\prime }Z^{\prime
}\Omega ^{-1}\mathcal{I}\left( \mathcal{I}^{\prime }\Omega ^{-1}\mathcal{I}%
\right) ^{-1}\mathcal{I}^{\prime }\Omega ^{-1}Z\text{ }\right] C
\end{eqnarray*}

Therefore, 
\begin{eqnarray}
\widehat{\beta } &=&CQC^{\prime }Z^{\prime }\Omega ^{-1}\left[ I-\mathcal{I}%
\left( \mathcal{I}^{\prime }\Omega ^{-1}\mathcal{I}\right) ^{-1}\mathcal{I}%
^{\prime }\Omega ^{-1}\right] Y_{T}  \notag \\
&&+\left[ I-CQC^{\prime }Z^{\prime }\Omega ^{-1}Z+CQC^{\prime }Z^{\prime
}\Omega ^{-1}\mathcal{I}\left( \mathcal{I}^{\prime }\Omega ^{-1}\mathcal{I}%
\right) ^{-1}\mathcal{I}^{\prime }\Omega ^{-1}Z\text{ }\right] Cb_{0}^{\ast }
\notag \\
&=&Cb_{0}^{\ast }+CQC^{\prime }Z^{\prime }\Omega ^{-1}\left( Y_{T}-\mathcal{I%
}\widehat{v}-ZCb_{0}^{\ast }\right) .  \label{beta_}
\end{eqnarray}

\begin{itemize}
\item The means squared error matrix
\begin{equation*}
Var\left( \beta |Y_{T}\right) =CQC^{\prime }-CQC^{\prime }Z^{\prime }\Omega
^{-1}ZCQC^{\prime }
\end{equation*}%
From (\ref{GLS}) and Lemma, we can show%
\begin{eqnarray*}
Var\left( \widehat{\beta }\right) &=&G^{-1}+G^{-1}EF^{-1}BG^{-1} \\
&=&\left( CQC^{\prime }-CQC^{\prime }Z^{\prime }\Omega ^{-1}ZCQC^{\prime
}\right) +G^{-1}Z^{\prime }H^{-1}\mathcal{I}\left( \mathcal{I}^{\prime
}\Omega ^{-1}\mathcal{I}\right) ^{-1}\mathcal{I}^{\prime }H^{-1}ZG^{-1}\text{%
\ } \\
&=&CQC^{\prime }-CQC^{\prime }Z^{\prime }\Omega ^{-1}ZCQC \\
&&+\left( CQC^{\prime }-CQC^{\prime }Z^{\prime }\Omega ^{-1}ZCQC^{\prime
}\right) Z^{\prime }H^{-1}\mathcal{I}\left( \mathcal{I}^{\prime }\Omega ^{-1}%
\mathcal{I}\right) ^{-1}\mathcal{I}^{\prime }H^{-1}Z \\
&&\times \left( CQC^{\prime }-CQC^{\prime }Z^{\prime }\Omega
^{-1}ZCQC^{\prime }\right) \\
&=&CQC^{\prime }-CQC^{\prime }Z^{\prime }\Omega ^{-1}ZCQC \\
&&+\left( CQC^{\prime }Z^{\prime }H^{-1}-CQC^{\prime }Z^{\prime }\Omega
^{-1}ZCQC^{\prime }Z^{\prime }H^{-1}\right) \mathcal{I}\left( \mathcal{I}%
^{\prime }\Omega ^{-1}\mathcal{I}\right) ^{-1}\mathcal{I}^{\prime } \\
&&\times \left( H^{-1}ZCQC^{\prime }-H^{-1}ZCQC^{\prime }Z^{\prime }\Omega
^{-1}ZCQC^{\prime }\right) \\
&=&CQC^{\prime }-CQC^{\prime }Z^{\prime }\Omega ^{-1}ZCQC \\
&&+\left( CQC^{\prime }Z^{\prime }H^{-1}-CQC^{\prime }Z^{\prime }\Omega
^{-1}\left( \Omega -H\right) H^{-1}\right) \mathcal{I}\left( \mathcal{I}%
^{\prime }\Omega ^{-1}\mathcal{I}\right) ^{-1}\mathcal{I}^{\prime } \\
&&\times \left( H^{-1}ZCQC^{\prime }-H^{-1}\left( \Omega -H\right) ^{\prime
}\Omega ^{-1}ZCQC^{\prime }\right) \\
&=&CQC^{\prime }-CQC^{\prime }Z^{\prime }\Omega ^{-1}ZCQC \\
&&+CQC^{\prime }Z^{\prime }\Omega ^{-1}\mathcal{I}\left( \mathcal{I}^{\prime
}\Omega ^{-1}\mathcal{I}\right) ^{-1}\mathcal{I}^{\prime }\Omega
^{-1}ZCQC^{\prime }.
\end{eqnarray*}

\item Note also that 
\begin{eqnarray*}
-F^{-1}BG^{-1} &=&-\left( \mathcal{I}^{\prime }\Omega ^{-1}\mathcal{I}%
\right) ^{-1}\mathcal{I}^{\prime }H^{-1}Z\left( CQC^{\prime }-CQC^{\prime
}Z^{\prime }\Omega ^{-1}ZCQC^{\prime }\right) \\
&=&-\left( \mathcal{I}^{\prime }\Omega ^{-1}\mathcal{I}\right) ^{-1}\mathcal{%
I}^{\prime }\left( H^{-1}ZCQC^{\prime }-H^{-1}ZCQC^{\prime }Z^{\prime
}\Omega ^{-1}ZCQC^{\prime }\right) \\
&=&-\left( \mathcal{I}^{\prime }\Omega ^{-1}\mathcal{I}\right) ^{-1}\mathcal{%
I}^{\prime }\left( H^{-1}ZCQC^{\prime }-H^{-1}\left( \Omega -H\right) \Omega
^{-1}ZCQC^{\prime }\right) \\
&=&-\left( \mathcal{I}^{\prime }\Omega ^{-1}\mathcal{I}\right) ^{-1}\mathcal{%
I}^{\prime }\Omega ^{-1}ZCQC^{\prime }
\end{eqnarray*}
\begin{eqnarray*}
-G^{-1}EF^{-1} &=&-\left( CQC^{\prime }-CQC^{\prime }Z^{\prime }\Omega
^{-1}ZCQC^{\prime }\right) Z^{\prime }H^{-1}\mathcal{I}\left( \mathcal{I}%
^{\prime }\Omega ^{-1}\mathcal{I}\right) ^{-1} \\
&=&-\left( CQC^{\prime }Z^{\prime }H^{-1}-CQC^{\prime }Z^{\prime }\Omega
^{-1}ZCQC^{\prime }Z^{\prime }H^{-1}\right) \mathcal{I}\left( \mathcal{I}%
^{\prime }\Omega ^{-1}\mathcal{I}\right) ^{-1} \\
&=&-\left( CQC^{\prime }Z^{\prime }H^{-1}-CQC^{\prime }Z^{\prime }\Omega
^{-1}\left( \Omega -H\right) H^{-1}\right) \mathcal{I}\left( \mathcal{I}%
^{\prime }\Omega ^{-1}\mathcal{I}\right) ^{-1} \\
&=&-CQC^{\prime }Z^{\prime }\Omega ^{-1}\mathcal{I}\left( \mathcal{I}%
^{\prime }\Omega ^{-1}\mathcal{I}\right) ^{-1}
\end{eqnarray*}

\item Therefore, 
\begin{equation*}
Var\left( \beta |Y_{T}\right) =Var\left( \widehat{\beta }\right) -Cov\left( 
\widehat{\beta },\widehat{v}\right) Var\left( \widehat{v}\right)
^{-1}Cov\left( \widehat{\beta },\widehat{v}\right) ^{\prime }.
\end{equation*}
\end{itemize}

\section*{Appendix 2: TV-VAR(2) with Time-Varying Intercepts}

\subsubsection*{VAR(2) Case: $p=2$ (i.e., 2 lags) and $k=3$ (i.e., 3 variables)}

To make the matrix $Z$, first define

\begin{eqnarray*}
Z_{t} &=&\underset{k\times \left( pk+1\right) k}{\underbrace{\left( \left[
1,y_{t-1}^{\prime },y_{t-2}^{\prime }\right] \otimes I_{k}\right) }} \\
&=&\left( \left[ 1,y_{t-1}^{\prime },y_{t-2}^{\prime }\right] \otimes
I_{k}\right) .
\end{eqnarray*}%
Then,

\begin{eqnarray*}
\text{\ }\underset{k\left( T-p\right) \times \left( pk+1\right) k\left(
T-p\right) }{\underbrace{Z}} &=&\left[ 
\begin{array}{cccc}
Z_{3} &  &  & 0 \\ 
& Z_{4} &  &  \\ 
&  & \ddots &  \\ 
0 &  &  & Z_{T}%
\end{array}%
\right] \\
&=&\left[ 
\begin{array}{cccc}
\left[ 1,y_{2}^{\prime },y_{1}^{\prime }\right] \otimes I_{k} &  &  & 0 \\ 
& \left[ 1,y_{3}^{\prime },y_{2}^{\prime }\right] \otimes I_{k} &  &  \\ 
&  & \ddots &  \\ 
0 &  &  & \left[ 1,y_{T-1}^{\prime },y_{T-2}^{\prime }\right] \otimes I_{k}%
\end{array}%
\right] \\
&=&z\otimes I_{k}
\end{eqnarray*}%
where

\begin{equation*}
\underset{\left( T-p\right) \times \left( kp+1\right) \left( T-p\right) }{%
\underbrace{z}}=\left[ 
\begin{array}{cccc}
\left[ 1,y_{2}^{\prime },y_{1}^{\prime }\right] &  &  & 0 \\ 
& \left[ 1,y_{3}^{\prime },y_{2}^{\prime }\right] &  &  \\ 
&  & \ddots &  \\ 
0 &  &  & \left[ 1,y_{T-1}^{\prime },y_{T-2}^{\prime }\right]%
\end{array}%
\right] .
\end{equation*}%
For the regression:

\begin{equation*}
\underset{k\left( T-p\right) \left( kp+2\right) \times 1}{\underbrace{\left[ 
\begin{array}{c}
Y_{T} \\ 
-b_{0}^{\ast }%
\end{array}%
\right] }}=\underset{k\left( T-p\right) \left( kp+2\right) \times \left(
kp+1\right) k\left( T-p\right) }{\underbrace{\left[ 
\begin{array}{c}
Z \\ 
-C^{-1}%
\end{array}%
\right] }}\underset{\left( T-p\right) k\left( kp+1\right) \times 1}{%
\underbrace{\beta ^{\ast }}}+\underset{k\left( T-p\right) \left( kp+2\right)
\times 1}{\underbrace{\left[ 
\begin{array}{c}
\varepsilon \\ 
\eta%
\end{array}%
\right] }},
\end{equation*}%
one needs to define

\begin{equation*}
X=\left[ 
\begin{array}{c}
Z \\ 
-C^{-1}%
\end{array}%
\right] .
\end{equation*}

Then,

\begin{eqnarray*}
\underset{\left( kp+1\right) k\left( T-p\right) \times \left( kp+1\right)
k\left( T-p\right) }{\underbrace{X^{\prime }X}} &=&\left[ 
\begin{array}{cc}
Z^{\prime } & -C^{-1\prime }%
\end{array}%
\right] \left[ 
\begin{array}{c}
Z \\ 
-C^{-1}%
\end{array}%
\right] \\
&=&\left[ Z^{\prime }Z+C^{-1\prime }C^{-1}\right] ,
\end{eqnarray*}

\ where

\begin{equation*}
\underset{\left( kp+1\right) k\left( T-p\right) \times \left( kp+1\right)
k\left( T-p\right) }{\underbrace{C}}=\left[ 
\begin{array}{cccc}
I & 0 & \cdots & 0 \\ 
I & I &  & \vdots \\ 
\vdots & \vdots & \ddots & 0 \\ 
I & I & I & I%
\end{array}%
\right] =\underset{\left( T-p\right) \times \left( T-p\right) }{\underbrace{%
\left[ 
\begin{array}{cccc}
1 & 0 & \cdots & 0 \\ 
1 & 1 &  & \vdots \\ 
\vdots & \vdots & \ddots & 0 \\ 
1 & 1 & 1 & 1%
\end{array}%
\right] }}\otimes I_{\left( kp+1\right) k}=c\otimes I_{\left( kp+1\right) k}.
\end{equation*}%
Here, 
\begin{equation*}
\underset{\left( T-p\right) \times \left( T-p\right) }{\underbrace{c}}=\left[
\begin{array}{cccc}
1 & 0 & \cdots & 0 \\ 
1 & 1 &  & \vdots \\ 
\vdots & \vdots & \ddots & 0 \\ 
1 & 1 & 1 & 1%
\end{array}%
\right] .
\end{equation*}

The rest of the matrices needed for GLS are:%
\begin{eqnarray*}
\underset{\left( T-p\right) k\times 1}{\underbrace{Y_{T}}} &=&\left[ 
\begin{array}{c}
y_{p+1} \\ 
y_{p+2} \\ 
\vdots \\ 
y_{T}%
\end{array}%
\right] ,\text{ \ } \\
\underset{k\left( T-p\right) \times \left( pk+1\right) k\left( T-p\right) }{%
\underbrace{Z}} &=&\left[ 
\begin{array}{cccc}
Z_{p+1} &  &  & 0 \\ 
& Z_{p+2} &  &  \\ 
&  & \ddots &  \\ 
0 &  &  & Z_{T}%
\end{array}%
\right] ,\text{ \ } \\
\underset{\left( T-p\right) k\left( kp+1\right) \times 1}{\underbrace{\beta }%
} &=&\left[ 
\begin{array}{c}
\beta _{p+1} \\ 
\beta _{p+2} \\ 
\vdots \\ 
\beta _{T}%
\end{array}%
\right] ,\text{ }\underset{\left( T-p\right) k\times 1}{\text{\ }\underbrace{%
\varepsilon }}=\left[ 
\begin{array}{c}
\varepsilon _{p+1} \\ 
\varepsilon _{p+2} \\ 
\vdots \\ 
\varepsilon _{T}%
\end{array}%
\right] ,\text{ \ \ \ } \\
\underset{\left( T-p\right) k\left( kp+1\right) \times 1}{\underbrace{\eta }}
&=&\left[ 
\begin{array}{c}
\eta _{p+1} \\ 
\eta _{p+2} \\ 
\vdots \\ 
\eta _{T}%
\end{array}%
\right] ,\text{ } \\
\text{\ }\underset{\left( kp+1\right) kp\left( T-p\right) \times \left(
kp+1\right) k\left( T-p\right) }{\underbrace{C}} &=&\left[ 
\begin{array}{cccc}
I & 0 & \cdots & 0 \\ 
I & I &  & \vdots \\ 
\vdots & \vdots & \ddots & 0 \\ 
I & I & I & I%
\end{array}%
\right] ,\text{ } \\
\text{ \ \ \ \ }\underset{\left( T-p\right) k\left( kp+1\right) \times
\left( T-p\right) k\left( kp+1\right) }{\underbrace{Q}} &=&\left[ 
\begin{array}{cccc}
Q_{p+1} &  &  & 0 \\ 
& Q_{p+2} &  &  \\ 
&  &  &  \\ 
0 &  &  & Q_{T}%
\end{array}%
\right] ,\text{\ } \\
\underset{\left( T-p\right) k\times \left( T-p\right) k}{\underbrace{H}} &=&%
\left[ 
\begin{array}{cccc}
H_{p+1} &  &  & 0 \\ 
& H_{p+2} &  &  \\ 
&  & \ddots &  \\ 
0 &  &  & H_{T}%
\end{array}%
\right] ,
\end{eqnarray*}

\begin{equation*}
\underset{\left( T-p\right) k\left( kp+1\right) \times 1}{\text{\ }%
\underbrace{b_{0}^{\ast }}}=\left[ 
\begin{array}{c}
b_{0} \\ 
0 \\ 
\vdots \\ 
0%
\end{array}%
\right] ,\text{ \ }P_{0}^{\ast }=\left[ 
\begin{array}{cccc}
P_{0} & 0 &  & 0 \\ 
0 & 0 &  &  \\ 
\vdots &  & \ddots &  \\ 
0 & 0 & \cdots & 0%
\end{array}%
\right] .
\end{equation*}